\documentclass[final,leqno]{siamltex1213}

\usepackage[dvips]{epsfig}
\usepackage{graphicx} 
\usepackage{latexsym}
\usepackage{verbatim}
\usepackage{amsmath}
\usepackage{amssymb}
\usepackage[margin=10pt,font=small,labelfont=bf,
labelsep=endash]{caption}
\usepackage{subfig}
\usepackage{bbm}
\usepackage[crop=pdfcrop]{pstool}
\usepackage[boxed]{algorithm2e}
\usepackage{mathtools}
\newcommand{\InsertFig}[3]{
  \begin{figure}[!htbp]
    \begin{center}
      \leavevmode
      #1
      \caption{#2}
      #3
    \end{center}
  \end{figure}
}

\newtheorem{remark}[theorem]{Remark}
\newtheorem{example}[theorem]{Example}
\newtheorem{fact}{Fact}

\newcommand{\ceil}[1]{\ensuremath{\left\lceil #1 \right\rceil}}

\newcommand{\thmref}[1]{Theorem~\ref{#1}}

\newcommand{\lemref}[1]{Lemma~\ref{#1}}

\newcommand{\figref}[1]{Figure~\ref{#1}}

\newcommand{\abs}[1]{\ensuremath{\left| #1 \right|}}

\newcommand{\MR}[1]{\ensuremath{\text{MR}_{#1}}}

\newcommand{\R}{\mathbb{R}}

\newcommand{\de}{\partial}

\newcommand{\intm}[1]{\int\limits_{-1}^{0} #1~d\mu}
\newcommand{\intp}[1]{\int\limits_{0}^{1} #1~d\mu}

\newcommand{\psip}[1]{\psi^{\left( #1\right)}_+}
\newcommand{\psim}[1]{\psi^{\left( #1\right)}_-}
\newcommand{\psiv}[1]{\psi^{\left( #1\right)}}

\newcommand{\Phip}[1]{{\phi^{\left( #1\right)}_+}} 
\newcommand{\Phiv}[1]{{\phi^{\left( #1\right)}}} 
\newcommand{\Phim}[1]{{\phi^{\left( #1\right)}_-}}
\newcommand{\Phipm}[1]{{\phi^{\left( #1\right)}_\pm}}

\newcommand{\range}[1]{R \left( #1 \right)}
\newcommand{\support}[1]{\text{supp}\left( #1 \right)}
\newcommand{\bb}{\mathbf{b}}
\newcommand{\bw}{\mathbf{w}}

\newcommand{\psipm}[1]{\psi^{\left( #1\right)}_{\pm}}

\numberwithin{equation}{section}
\title{Higher order mixed moment approximations for the Fokker-Planck equation in one space dimension}

\author{
Florian Schneider\thanks{Fachbereich Mathematik, TU Kaiserslautern, Erwin-Schr\"odinger-Str., 67663 Kaiserslautern, Germany, {\tt fschneid@mathematik.uni-kl.de}}
\and
Graham Alldredge\thanks{Department of Mathematics, RWTH Aachen University, Schinkelstr. 2,52062 Aachen, Germany, {\tt Graham Alldredge <alldredge@mathcces.rwth-aachen.de>}}
\and
Martin Frank\thanks{Department of Mathematics, RWTH Aachen University, Schinkelstr. 2,52062 Aachen, Germany, {\tt Martin Frank <frank@mathcces.rwth-aachen.de>}}
\and
Axel Klar\thanks{Fachbereich Mathematik, TU Kaiserslautern  and  Fraunhofer Institut f\"ur Techno- und Wirtschaftsmathematik, 67663 Kaiserslautern, Germany, {\tt klar@itwm.fraunhofer.de}}}

\date{}

\begin{document}

\maketitle
\slugger{siap}{xxxx}{xx}{x}{x--x}%slugger should be set to mms, siap, sicomp, sicon, sidma, sima, simax, sinum, siopt, sisc, or sirev
\begin{abstract}
We study mixed-moment models (full zeroth moment, half higher moments) for a Fokker-Planck equation in one space dimension. Mixed-moment minimum-entropy models are known to overcome the zero net-flux problem of full-moment minimum entropy M${}_n$ models. Realizability theory for these mixed moments of arbitrary order is derived, as well as a new closure, which we refer to as Kershaw closures. They provide non-negative distribution functions combined with an analytical closure.
Numerical tests are performed with standard first-order finite volume schemes and compared with a finite-difference Fokker-Planck scheme.
\end{abstract}

\noindent

% {\bf Key words.}

\ifpdf
    \graphicspath{{Images/PNG/}{Images/PDF/}{Images/}}
\else
    \graphicspath{{Chapter4/Chapter4Figs/EPS/}{Chapter4/Chapter4Figs/}}
\fi

%%%%%%%%%%%%%%%%%%%%%%
\section{Introduction}
%%%%%%%%%%%%%%%%%%%%%%
%
In recent years many approaches have been proposed for the solution of time-dependent kinetic transport equations, for example in electron radiation therapy or radiative heat transfer problems. Since transport equations are very high-dimensional one tries to approximate them, for example, by so-called moment models. Typical representatives are polynomial P${}_n$-methods \cite{Jea17,Eddington,Brunner2005} and their simplifications, the SP${}_n$ \cite{Gel61} methods. These models are computationally inexpensive since they form an analytically closed system of hyperbolic differential equations. However, they suffer from severe drawbacks:
The P${}_n$ methods are generated by closing the  balance equations with a distribution function which is a polynomial in the angular variable. This implies that this distribution function might have negative values resulting in non-physical values like a negative density.  Additionally, in many cases a very high number of moments is needed for a reasonable approximation of the transport solution. This is particularly true in beam cases, where the exact transport solution forms a Dirac delta.

The entropy minimization M${}_n$-models \cite{Min78,DubFeu99,BruHol01,Monreal2008,AllHau12} are expected to overcome this problem since the moment equations are always closed with a positive distribution. In many situations these models perform very well, but they can produce unphysical steady-state shocks due to a zero net-flux problem. It has been shown by Hauck that these shocks exist for every odd order \cite{Hauck2010}. Perturbation theory has been applied to the M${}_1$ closure which, in some cases, removes the drawback of shocks \cite{Frank:1471763}.
 To improve this situation, half- or partial-moment models were introduced in \cite{DubKla02,FraDubKla04}. These models work especially well in one space dimension, because they capture the potential discontinuity of the probability density in the angular variable which in 1D is well-located. If however, a Fokker-Planck operator is used instead of the standard integral scattering operator, so that scattering is extremely forward-peaked \cite{Pom92}, these half-moment approximations do not work for reasons we will explain below.
 
To improve this situation a new model with mixed moments was proposed in \cite{Frank07} which is able to avoid this problem. Instead of choosing  full or half moments, a mixed-moment approximation is used. Contrary to a typical half moment approximation, the lowest order moment (density) is kept as a full moment while all higher moments are half moments. This ensures the continuity of the underlying distribution function while maintaining the flexibility of the half moment approach.

The fact that this modification of the moment model gives reasonable results encourages us to study the general case for an arbitrary number of moments. A problem that was left open in the previous work \cite{Frank07} was the question of realizability of these mixed moments. We will define realizability precisely below. 

The first part of the  present paper deals with the investigation of necessary and sufficient realizability conditions for the mixed moment problem of arbitrary order. We note that in contrast to the full moment problem, the  mixed moment problem has not been discussed in literature up to now.\\
 In the second part of the paper, higher-order mixed-moment closures of the balance equations are derived, investigated and compared to each other, continuing the work done in \cite{Frank07} where mixed moments of order one were introduced. These closures, which we call Kershaw closures because they are inspired by ideas outlined in \cite{Ker76}, give an explicit, analytical closure for the moment system and are therefore very efficient for higher orders since no nonlinear system has to be solved. The paper is organized as follows.\\
  In Section~\ref{momapprox}, we give an overview of the basic Fokker-Planck equation and the different moment models which can be derived from it.  In particular, half-moment approximations and mixed half-full moment approximations are discussed. Section \ref{mixed} contains the necessary and sufficient conditions for realizability of mixed moments of arbitrary order. Several higher-order mixed-moment closures of the balance equations including polynomial, minimum-entropy and Kershaw closures are developed in Section \ref{closure}.
The  resulting set of equations are numerically compared with each other and the Fokker-Planck equation in Section~\ref{sec:SIM}. We conclude with a discussion of the results and present open problems for future work in Section~\ref{sec:CON}.

%%%%%%%%%%%%%%%%%%%%%%
\section{Moment approximations and realizability}
\label{momapprox}
We consider the one-\\dimensional Fokker-Planck equation for $t>0$, $x \in \R$, $\mu \in [-1,1]$
\begin{align}
\label{eq:FokkerPlanck1D}
\partial_t\psi(t,x,\mu) + \mu \partial_x  \psi(t,x,\mu) + \sigma_a(t,x) \psi(t,x,\mu) = \frac{T(t,x)}{2}\Delta_\mu \psi(t,x,\mu) + Q(t,x,\mu),
\end{align}
where $\sigma_a$ is the absorption coefficient, $T$ is the so-called transport coefficient, $Q$ a source and $\Delta_\mu$ is given by the one-dimensional projection of the Laplace-Beltrami operator on the unit sphere
\begin{align}
\Delta_\mu \psi = \frac{\partial}{\partial\mu}\left(\left(1-\mu^2\right) \frac{\partial \psi}{\partial \mu}\right).
\end{align}
Equation \eqref{eq:FokkerPlanck1D} can be derived from the standard linear transport equation in the limit of forward-peaked anisotropic scattering \cite{Pom92}. The transport coefficient is the difference of the zeroth and the first moment of the scattering kernel.

We define the moments of the probability density $\psi$ as
$$
\psi^{(j)} = \int_{-1}^1 \mu^j \psi d\mu,~~j\geq 0.
$$
We multiply (\ref{eq:FokkerPlanck1D}) by $1$ and $\mu$, respectively, integrate over $[-1,1]$ and obtain
\begin{align}
	 \partial_t \psi^{(0)} + \partial_x \psi^{(1)} + \sigma_a \psi^{(0)}  &= Q^{(0)}\\
	 \partial_t \psi^{(1)} + \partial_x \psi^{(2)}   + \sigma_a \psi^{(1)} &= - \frac{T}{2}\psi^{(1)} + Q^{(1)} .
\end{align}
This is a system of moment equations, which could in principle be continued. If however, we want to truncate this system at a finite order we see that the system is underdetermined (two equations for the three unknowns $\psi^{(0)}$, $\psi^{(1)}$, $\psi^{(2)}$). Thus we need to devise an additional relation which is called moment closure. Typically, this is done by assuming a specific form for the density, known as an \emph{ansatz}, which depends on the known moments. For instance, we could assume that $\psi$ is a linear function in $\mu$. It matches the known moments $\psi^{(0)}$ and $\psi^{(1)}$ if
\begin{equation}
\label{eq:P1}
\psi_{P_1}(\mu) = \frac12 \psi^{(0)} + \frac32 \psi^{(1)} \mu.
\end{equation}
This distribution can be integrated to obtain $\psi^{(2)}$ in terms of $\psi^{(0)}$ and $\psi^{(1)}$.

When dealing with closures, it is always crucial to ask which moments can be realized by a physical (that is, nonnegative) distribution. In the case of two moments, we have as a necessary condition that
\begin{equation}\label{eq:M1Rrealizability}
	\psi^{(0)}=\int_{-1}^1  \psi d\mu\geq 0 \quad\text{and}\quad |\psi^{(1)}|=|\int_{-1}^1 \mu \psi d\mu| \leq \psi^{(0)}.\end{equation}
Here these conditions are also sufficient, that is, for each pair of moments $(\psi^{(0)},\psi^{(1)})$ that satisfies (\ref{eq:M1Rrealizability}) one can find a nonnegative density that reproduces these moments. This is the realizability problem mentioned in the introduction. (Note, that realizability of course does not guarantee the positivity of any representing distribution:  the linear closure (\ref{eq:P1}) can become negative even when $\psi^{(0)}$ and $\psi^{(1)}$ satisfy \eqref{eq:M1Rrealizability}.)

This paragraph follows the definitions and results from \cite{CurFial91}. For more details on the original truncated Hausdorff moment problem see e.g. \cite{curto1996solution,Curto1997,Laurent04revisitingtwo}.
%\begin{definition}
%Given an infinite sequence of real numbers $\gamma = \left\{\gamma_0,\gamma_1,\ldots \right\}$, an interval $[a,b]\subset\mathbb{R}$, the \textbf{K-power Hausdorff moment problem with data} $\gamma$ entails finding a positive Borel measure $\mu$ on $\mathbb{R}$ such that
%\begin{align*}
%\int t^j ~d\mu(t) = \gamma_j ~~~~0\leq j \leq \infty
%\end{align*}
%and $supp~\mu\subseteq [a,b]$
%\end{definition}
%\begin{remark}
%This can be translated in the context of radiation transport:
%Find a {non-negative distribution function} $\psi\geq 0$ such that
%\begin{align*}
%\int_{-1}^1 \mu^j \psi(\mu) ~d\mu = \psiv{j} ~~~~0\leq j \leq \infty
%\end{align*}
%\end{remark}
%We are especially interested in the so called \emph{truncated K-power Hausdorff moment problem} which is given by:
\begin{definition}
Given a vector of real numbers $\left( \psiv{0},\psiv{1},\ldots,\psiv{n} \right)$ and an interval $[a,b]\subset\mathbb{R}$, the \emph{truncated Hausdorff moment problem} entails finding a positive Borel measure $\lambda$ on $\mathbb{R}$ such that
\begin{align}
\label{eq:TruncMoment}
\int \mu^j ~d\lambda(\mu) = \psiv{j} ~~~~0\leq j \leq n
\end{align}
and $\rm{supp}~\lambda\subseteq [a,b]$.

We also define the \emph{realizability domain} $\mathcal{R}_n$ as the set of vectors $\left(\psiv{0},\psiv{1},\ldots,\psiv{n} \right) \in \mathbb{R}^{n+1}$ such that the truncated moment problem \eqref{eq:TruncMoment} has a solution for $[a,b] = [-1, 1]$.
\end{definition}
\begin{remark}
From now on we assume that $\lambda$ has a distribution so that the truncated moment problem can be translated to:
Find a {non-negative distribution} $\psi(\mu)\geq 0$ such that
\begin{align*}
\int_a^b \mu^j \psi(\mu) ~d\mu = \psiv{j}~~~~0\leq j \leq n.
\end{align*}
Such a density $\psi$ is said to be a \emph{representing distribution} of the moment vector $\left( \psiv{0},\psiv{1},\ldots,\psiv{n} \right)$.

We allow (at least formally) the case when $\psi$ includes a finite linear combination of Dirac $\delta$-functions.  Distributions $\psi$ which consist only of a linear combination of Dirac $\delta$-functions are called \emph{atomic}.
%In fact, it can be shown that in case of solvability a finite atomic measure can be found which at least formally can be written as having a density consisting of a linear combination of Dirac $\delta$ functions.
\end{remark}

\begin{definition}
In the following, the inequality $A\geq 0$ indicates that a matrix $A$ is positive semidefinite. For hermitian matrices $A,B\in\R^{n\times n}$ we define the partial ordering $"\geq"$ by $A\geq B$ iff $A-B\geq 0$, that is $A-B$ is positive semi-definite.
\end{definition}

The main result is the following:
\begin{theorem}
\label{thm:FullMomentRealizability}
%For the truncated Hausdorff moment problem, the following realizability conditions hold.
Define the Hankel matrices
$$
A(k):=\left(\psiv{i+j}\right)_{i,j=0}^k, \quad B(k):=\left(\psiv{i+j+1}\right)_{i,j=0}^k, \quad C(k):=\left(\psiv{i+j}\right)_{i,j=1}^k.
$$
Then the truncated Hausdorff moment problem has a solution if and only if
\begin{itemize}
\item for $n=2k+1$,
\begin{align}
%A(k)\geq 0 \qquad \text{and} \\
bA(k)\geq B(k)\geq aA(k); \label{eq:haus-odd}
\end{align}
\item for $n=2k$,
\begin{align}
A(k) &\geq 0 \qquad \text{and} \label{eq:haus-even-1} \\
(a+b)B(k-1) &\geq abA(k-1) + C(k). \label{eq:haus-even-2}
\end{align} 
\end{itemize}
\end{theorem}
\begin{proof}
See \cite{CurFial91} for a detailed proof of this theorem.
\end{proof}

The following remark will be important later on.
\begin{remark}
\label{rem:GeneratingFunction} In \cite{CurFial91} the authors showed that, starting from the previous theorem, there exists a minimal atomic representing distribution $\psi$ (in the sense that it contains the fewest possible number of Dirac $\delta$ functions while still representing the moments) and that one can directly find this distribution with the help of its \emph{generating function}. This is the consequence of a recursiveness property of the Hankel matrices $A(k)$ and $B(k)$.\\
Let $\Phi = \left(\varphi_0,\ldots,\varphi_{r-1}\right):= A(r-1)^{-1}v(r,r-1)$, where $v(i,j) = \left(\psiv{i+l}\right)_{l=0}^j$ and $r$ is the smallest integer such that $A(r - 1)$ is nonsingular,%
\footnote{
This integer $r$ is the \emph{Hankel rank} defined in \cite{CurFial91}.
} %
the generating function is defined by
\begin{align*}
g_\varphi(\mu) = \mu^r-\sum\limits_{i=0}^{r-1}\varphi_i\mu^i
\end{align*}
The roots of this polynomial give the atoms $\mu_i$ of the distribution $\psi$ whereas the densities are calculated afterwards from the Vandermonde system 
\begin{align*}
\rho_0\mu_0^i+\cdots +\rho_{r-1}\mu_{r-1}^i = v(i,0)~~~i=0\ldots r-1.
\end{align*}
Furthermore, when the moment vector $\left(\psiv{0},\psiv{1},\ldots,\psiv{n} \right)$ is on the boundary of the realizability domain (which is equivalent to $r < k + 1$, where $k$ is as defined in \thmref{thm:FullMomentRealizability}), the minimal atomic representing measure is the unique representing measure.

Note that finding the atoms can be done in a robust way. Although the condition of the Vandermonde system is very bad (namely exponential in $r$ \cite{Gautschi1987}), efficient and robust algorithms for the solution of this linear system exist \cite{Bjorck1970}. In our paper we will only provide results where the solution of this system can be obtained analytically using symbolic calculations.
\end{remark}

%%%%%%%%%%%%%%%%%%%%%%%%%%%%%%%%%%%
\section{Moment closures}
Several moment approximations have been extensively studied. Here we briefly recall some of them and also add a new closure strategy in Section \ref{sec:kershaw}.

We first remark that all closures provided here have benefits and drawbacks. The polynomial closure P${}_n$ (Section \ref{sec:Pn}) as well as minimum-entropy closures (for full, half and mixed moments, Sections \ref{sec:FMMinimumEntropy}, \ref{sec:HMMinimumEntropy} and \ref{sec:NEW}, respectively) are strictly hyperbolic. This is not the case for QMOM, which we describe in detail in Section \ref{sec:QMOM}. For arbitrary-order Kershaw closures (Section \ref{sec:kershaw}) it is not known if the system is strictly hyperbolic. 
Higher-order minimum-entropy models provide a good approximation of the underlying kinetic problem but suffer from high numerical costs for the inversion of the moment problem. On the other hand they can be stated and implemented and can preserve realizability without the knowledge of the realizability conditions. If one knows the realizability conditions, the numerical costs can be avoided using corresponding Kershaw closures which are based on analytical closures of the higher moments without the need of numerical inversion. They are therefore comparable in speed with P${}_n$ models while having the same realizability domain as the true solution. However this is (especially in more than one spatial dimension) not yet always possible.

\subsection{P${}_n$-equations}
\label{sec:Pn}
%%%%%%%%%%%%%%%%%%%%%%%%%%%%%%%%%%%
%
The $P_n$ equations can be most easily understood as a \\Galerkin semi-discretization in the angle $\mu$. We expand the distribution function in terms of a truncated series expansion,
\begin{equation}
	\psi_{P_N}(t,x,\mu) = \sum_{l=0}^N \psi^{(l)}(t,x) \frac{2l+1}{2} P_l(\mu),
\end{equation}
where P${}_l$ are the Legendre polynomials. The expansion  coefficients are the moments with respect to the Legendre polynomials:
\begin{equation}
	 \psi^{(l)}(t,x) = \int_{-1}^1 \psi(t,x,\mu) P_l(\mu)d\mu.
\end{equation}
Note that we use the index $l$ to distinguish these moments from the moments taken with respect to the monomials.
The Galerkin projection is done by testing the Fokker-Planck equation with $P_l$, integrating with respect to $\mu$ over $[-1,1]$ and using the recursion relation of the Legendre polynomials to obtain the P${}_N$ equations
\begin{equation}
%\begin{split}
	\partial_t \psi^{(l)} + \partial_x\left( \frac{l+1}{2l+1} \psi^{(l+1)}+\frac{l}{2l+1}\psi^{(l-1)}\right)
	+ \sigma_a \psi^{(l)} = - \frac{T}{2} l(l+1) \psi^{(l)} + Q^{(l)}
%\end{split}
\end{equation}
for $l=0,\dots,N$. Recall that $\psi^{(N+1)}=0$. We have also used the fact that the Legendre polynomials are eigenfunctions of the Laplace-Beltrami operator.

%%%%%%%%%%%%%%%%%%%%%%%%%%%%
\subsection{Minimum Entropy Closure}
\label{sec:FMMinimumEntropy}
%%%%%%%%%%%%%%%%%%%%%%%%%%%%
%
The approximations based on the expansion of the distribution function into a polynomial suffer from several drawbacks~\cite{DubFeu99}. As mentioned above, the distribution function can become negative and thus the moments computed from this distribution can become unphysical. One way to overcome this problem is to use an entropy minimization principle to obtain the constitutive equation to close the moment equations. The first step to derive the minimum entropy (M${}_1$) model~\cite{DubFeu99,AniPenSam91,BruHol01} is identical to the P${}_n$ method. We test with the monomials and get 
\begin{align}
\partial_t \psi^{(0)}+\partial_x \psi^{(1)} +  \sigma_a \psi^{(0)} &= Q^{(0)}\\
\partial_t \psi^{(1)}+\partial_x \psi^{(2)} +  \sigma_a \psi^{(1)}  &= - T \psi^{(1)} + Q^{(1)}.
\end{align}
Now instead of using the Galerkin ansatz we determine a distribution function $\psi_{M_1}$ that minimizes a functional, the entropy $H$, under the constraint that it reproduces the lower order moments,
\begin{equation}
\label{eq:M1Moments}
  \int_{-1}^1 \psi_{M_1}d\mu = \psi^{(0)} \quad\textrm{and}\quad \int_{-1}^1 \mu\psi_{M_1}d\mu = \psi^{(1)}.
\end{equation}
If the entropy is
$$
H(\psi) = \int_{-1}^1 \psi\log\psi -\psi d\mu
$$
then the minimizer can formally be written as \cite{Min78}
\begin{equation}
	\psi_{M_1}(\mu) = e^{\alpha + \beta \mu}.
\end{equation}
This is a Maxwell-Boltzmann type distribution and $\alpha$, $\beta$ are Lagrange multipliers enforcing the constraints.

It is not possible to express the highest moment $\psi^{(2)}$ explicitly in terms of $\psi^{(0)}$ and $\psi^{(1)}$, but we can write the flux function as
\begin{equation}
 	\psi^{(2)} = \chi\left(\frac{\psi^{(1)}}{\psi^{(0)}}\right)\psi^{(0)}.
\end{equation}
The so-called Eddington factor $\chi$ can be computed numerically, see \cite{Min78,BruHol01}.

%For boson entropy one can even find a closed form for the Eddington factor \cite{DubFeu99}:
%\begin{align*}
%\chi = \cfrac{5-2\sqrt{4-3\left(\cfrac{\psi^{(1)}}{\psi^{(0)}}\right)^2}}{3}
%\end{align*}

%%%%%%%%%%%%%%%%%%%%%%%%%%%%
\subsection{Kershaw Closure}\label{sec:kershaw}
%%%%%%%%%%%%%%%%%%%%%%%%%%%%
%
We now want to develop another closure strategy which we call Kershaw closure. The key idea of Kershaw in \cite{Ker76} is to derive a closure which preserves the realizability conditions.  First we recall that on the boundary of the realizability domain the distribution function must be atomic, that is a linear combination of Dirac delta functions. Since the realizable set $\mathcal{R}_n$ is a convex cone one can linearly interpolate between the boundary distributions to find a solution which realizes all moments. The resulting distribution is therefore exact on the boundary of the realizable set. We additionally want it to be exact in the equilibrium limit, where $\psi$ is independent on the angular variable.

%To introduce this ansatz, we use the realizability results from \cite{Monreal}.
We use the $n = 2$ case to introduce this method.  Here the realizability conditions are 
$$
\psi^{(0)}\geq 0,\quad -\psi^{(0)} \leq \psi^{(1)}\leq \psi^{(0)},\quad (\psi^{(1)})^2\leq \psi^{(0)}\psi^{(2)}\leq \psi^{(0)}.
$$
If we define the normalized moments $\phi^{(j)} := \cfrac{\psi^{(j)} }{ \psi^{(0)}}$ the above conditions can be rewritten as
$$
-1 \leq \phi^{(1)}\leq 1,\quad (\phi^{(1)})^2\leq \phi^{(2)}\leq 1.
$$
Thus the relative moment $\phi^{(2)}$ is bounded from above and below by two functions depending on $\phi^{(1)}$.
The distribution on the lower second order realizability boundary ($\phi^{(2)} = (\phi^{(1)})^2$) is given by 
\begin{align*}
\psi_{low} = \psiv{0}\delta\left(\phi^{(1)}-\mu\right)
\end{align*}
while the upper second order realizability boundary distribution ($\phi^{(2)}=1$) is given by
\begin{align*}
\psi_{up} = \psiv{0}\left(\frac{1-\phi^{(1)}}{2}\delta(1+\mu) + \frac{1+\phi^{(1)}}{2}\delta(1-\mu)\right)
\end{align*} 
By linearity of the problem, every convex combination $\psi_{K_1} = \alpha\psi_{low} + (1-\alpha)\psi_{up}$, $\alpha\in[0,1]$ will reproduce all first moments and satisfies $\psi_{K_1}\geq 0$.  We choose $\alpha$ so that our closure is exact for moments of the constant distribution $\psi_{const} \equiv \psiv{0}/2$. Calculating moments up to order $2$ of $\psi_{const}$ gives normalized moments $(\phi^{(1)},\phi^{(2)}) = (0,\frac{1}{3})$. Plugging this into the equations above we end up with
\begin{align*}
\phi^{(2)} = \alpha (\phi^{(1)})^2+(1-\alpha) \stackrel{\phi^{(1)}=0}{=} (1-\alpha)\cdot 1 \stackrel{!}{=} \frac{1}{3}
\end{align*}
This implies $\alpha = \frac{2}{3}$. Altogether we obtain
$$
\phi^{(2)} = \frac13(2(\phi^{(1)})^2+1).
$$
This relation is demonstrated in \figref{fig:K1example}. The value of $\phi^{(2)}$ on the upper boundary (dashed black line), $\phi^{(2)}_{up} = 1$, is convexly combined with its value on the lower boundary (blue solid line),  $\phi^{(2)}_{low} = \left(\Phiv{1}\right)^2$, to obtain the realizable convex-combination (red dashed-dotted line) which interpolates the equilibrium point $\left(0,\frac13\right)$ (red point).
\InsertFig{
         \includegraphics[scale=0.5]{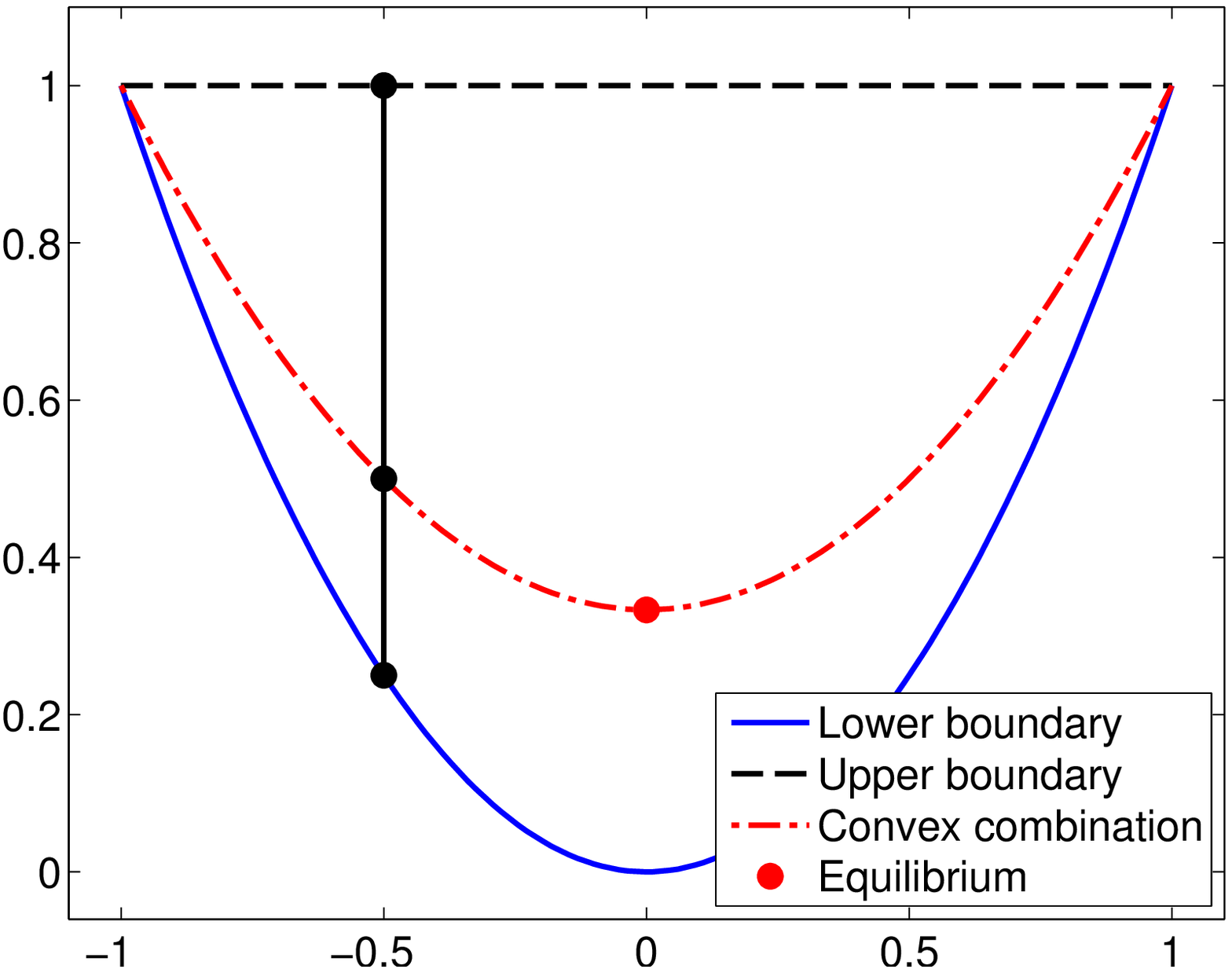}
         
}{Construction of the K${}_1$ closure using a convex combination of the upper (dashed black) and lower (blue solid) boundary second moments.}{\label{fig:K1example}}

We call this model the Kershaw $K_1$ closure. Note that using this approach it is not possible to provide point-values of the distribution. It is, however, possible to derive this class of models without the explicit representation of the distributions on the upper and lower boundaries \cite{Monreal}.

%%%%%%%%%%%%%%%%%%%%%%%%%%%%
\subsection{Quadrature method of moments}
\label{sec:QMOM}
%%%%%%%%%%%%%%%%%%%%%%%%%%%%
         
Similar ideas are also known in hydrodynamics and other fields as the Quadrature Method of Moments (QMOM) \cite{McG97,DesFoxVil08,Fox08,Fox09,YuaLauFox12}. These ideas have been recently applied to radiative transfer \cite{VikHauWanFox13}. The general idea is to use an $N$-atomic discrete measure 
$$
\psi_{QMOM_N} = \sum\limits_{i=1}^{N}\alpha_i\delta(\mu-\mu_i)
$$ 
where $N$ is a fixed number that may be independent of the order $n$. Typically one uses $n=2N-1$ to obtain $2N$ equations and $2N$ unknowns ($N$ weights and abscissas). Plugging $\psi_{QMOM_N}$ into the moment problem \eqref{eq:TruncMoment} yields a nonlinear system which can be solved using the so-called Wheeler-algorithm \cite{McGraw1997,VikHauWanFox13} which diagonalizes a tridiagonal matrix to find the weights and abscissas. This results in a robust and efficient algorithm for the inversion of the moment problem.
Additionally a major advantage of the QMOM approach is the fact that it can correctly reproduce moments which lie on the realizability boundary \cite{VikHauWanFox13} since there the distribution is (uniquely) atomic \cite{Ker76}. 
%Note that this is (at least numerically) not possible for minimum entropy models since the numerical quadratures which are necessary here \cite{AllHau12} will fail when the exponential function is steepened to a Dirac delta in the limit.

Although the original $N-$node QMOM, tracking $2N$ moments, is very efficient and has a lots of advantages it has also some drawbacks \cite{VikHauWanFox13}: The reconstructed higher moments will always lie on the (higher order) realizability boundary which immediately implies that it is not possible to correctly reproduce the equilibrium distribution. However, numerical experiments show \cite{VikHauWanFox13} that the equilibrium limit can be captured with EQMOM (Extended QMOM). Another result of $\psi_{QMOM_N}$ being on the $n+1$-th order realizability boundary is that the eigenvalues of the Jacobian of the flux are not unique giving only weak hyperbolicity of the system of partial differential equations. Minimum entropy models, on the other hand, exhibit strict hyperbolicity everywhere in the interior of the realizability domain. It is also not possible to obtain point-wise values of the distribution function itself. 

%If one allows more than one delta here it is possible to enforce this equilibrium interpolation as it is done for the Kershaw closures in Section \ref{sec:kershaw}.

%As an example consider again the M${}_1$ problem \eqref{eq:M1Moments}. The typical QMOM approach would be $\psi_{QMOM_1} = \alpha\delta(\mu-\tilde{\mu})$ where the node $\mu$ and weight $\alpha$ need to be chosen in such a way that \eqref{eq:M1Moments} is satisfied. Thus one solves the system
%\begin{align*}
%\alpha &= \psiv{0},~~~\alpha\tilde{\mu} = \psiv{1}
%\end{align*}
%Thus we have $\psi_{QMOM_1} = \psiv{0}\delta(\mu-\Phiv{1})$. But this gives in the closure $\Phiv{2} = \Phiv{1}^2$ which cannot interpolate the equilibrium point $(\phi^{(1)},\phi^{(2)}) = (0,\frac{1}{3})$.

%%%%%%%%%%%%%%%%%%%%%%%%%%%%
\subsection{Half Moment Approximation}
\label{sec:HMMinimumEntropy}
%%%%%%%%%%%%%%%%%%%%%%%%%%%%
%
A model which has been successfully applied to radiative transfer in one dimension, removing some drawbacks of the minimum entropy model, is the half-moment approximation~\cite{SchFraPin04,TurFraDubKla04}. A typical disadvantage of the minimum entropy solution can be seen in the numerical section, for example in Figure \ref{fig:2Beams1}. The idea of half-moment models is to average not over all directions but over certain subsets, for example the sets of particles moving left and those moving right. In one dimension, this means to integrate over $[-1,0]$ and $[0,1]$ respectively. We denote the half moments by
\begin{equation}
	\psi_+^{(j)}(t,x) = \int_0^1 \mu^j \psi(t,x,\mu) d\mu\quad\textrm{and}\quad \psi_-^{(j)}(t,x) = \int_{-1}^0 \mu^j\psi(t,x,\mu) d\mu.
\end{equation}
Applying this approach to the Fokker-Planck equation we obtain,
\begin{equation}
	\partial_t \psi_+^{(0)} + \partial_x \psi_+^{(1)} + \sigma_a \psi_+^{(0)} = \frac{T}{2} \int_0^1 \partial_\mu(1-\mu^2)\partial_\mu \psi d\mu  + Q_+^{(0)}
\end{equation}	
If we use integration by parts, the integral on the right hand side becomes
\begin{equation}
\int_0^1 \partial_\mu(1-\mu^2)\partial_\mu \psi d\mu = -\partial_\mu \psi(0^+).
\end{equation}
Similarly,
\begin{equation}
\int_{-1}^0 \partial_\mu(1-\mu^2)\partial_\mu \psi d\mu = \partial_\mu \psi(0^-).
\end{equation}
We note that, in contrast to the spherical harmonics approach, on the right hand side a microscopic term, namely a value of the distribution itself instead of its moments, appears. A naive approach would be to use entropy minimization on each half space separately. Thus we would use the discontinuous ansatz
\begin{equation}
	\psi_{HM_1}(\mu) = \begin{cases} e^{\alpha_-+\beta_- \mu} &\text{for}\quad \mu\in[-1,0]\\ e^{\alpha_++ \beta_+ \mu} &\text{for}\quad \mu\in[0,1]\end{cases}
\end{equation}
to close both the flux and the right hand side. This, however, would violate the important conservation property
$$
\int_{-1}^1 \partial_\mu(1-\mu^2)\partial_\mu \psi d\mu = 0
$$ 
of the Fokker-Planck equation resulting in a wrong approximation of the original equation \cite{Frank07}. One would therefore have to model the boundary terms differently. This is similar to the problems one encounters when the Discontinuous Galerkin method is applied to diffusion equations, see e.g. \cite{Bassi1997,Cockburn1998}. Similar interface terms appear that have to be approximated carefully. The deeper theoretical reason for these problems is that the domain of definition of the Laplace-Beltrami operator defines a continuous symmetric bilinear form $$
a(\lambda,\psi) :=-\int_{-1}^1 (1-\mu^2)\partial_\mu \lambda(\mu)\partial_\mu\psi d\mu $$ on the space $H^1([-1,1],\mathbb{R})$ \cite{Trofimov2002}. The method of moments is a Galerkin method, which should use subspaces of the domain of definition of the symmetric bilinear form as ansatz spaces. In 1D this excludes discontinuous functions. Instead of modeling the interface terms, we instead move to a conforming discretization by demanding continuity of the ansatz.

\subsection{A mixed moment method}
\label{sec:NEW}
%%%%%%%%%%%%%%%%%%%%%%%%%%%%%%%%%%%%%%%%%%%%%%%%%%%%%%%%%%%%%
%
%The half moment model was designed in the context of radiative transfer to remove the drawbacks of M${}_1$ and the P${}_n$ models.
%%For the case of equations with integral terms the Half Moment Minimum Entropy Model guarantees positivity, adapts its speed of propagation and eliminates the possibility of unphysical discontinuities as in the full moment minimum entropy model. The Half Moment $P_N$ model is  an improvement over the $P_N$ model but lacks the correct speed of propagation and does not guarantee positivity.
%However, in the present Fokker-Planck case, the numerical results show that the half moment model (both with polynomial and minimum entropy closure) fails dramatically \cite{Frank07}. One hint at this failure is the appearance of the terms $\partial_\mu \psi(0)$ in the derivation, which are depending on the distribution function and not directly on the moments of this function. A second indication that the Half Moment Model above is not a good approximation to the Fokker-Planck equation is that the ``+'' and ``-'' equations decouple. Thus the model cannot be a consistent discretization of the Fokker-Planck equation. The mathematical reason for the failure is that a discontinuous ansatz function does not reflect the smoothness of the solution, which is imposed by the angular diffusion operator. 

To impose continuity, we test \eqref{eq:FokkerPlanck1D} with $1$ and integrate over $[-1,1]$ and then with $\mu$ and integrate over $[-1,0]$ and $[0,1]$ separately. One obtains balance equations of the form
\begin{align}
	 \partial_t \psi^{(0)} + \partial_x (\psi_+^{(1)}+\psi_-^{(1)}) + \sigma_a \psi^{(0)}  &= Q^{(0)}\\
	 \partial_t \psi_+^{(1)} + \partial_x \psip{2}   + \sigma_a \psi_+^{(1)} &=\frac{T}{2}(\psi(0^+)-2\psi_+^{(1)}) + Q+^{(1)} \\
	 \partial_t \psi_-^{(1)} + \partial_x \psim{2}   + \sigma_a \psi_-^{(1)}&= \frac{T}{2}(-\psi(0^-)-2\psi_-^{(1)})+ Q_-^{(1)}.
\end{align}
We close this system with an underlying distribution that is continuous, because it is in the domain of definition of the Laplace-Beltrami operator, and will thus also preserve the conservation property of the Fokker-Planck equation. This happens automatically if we use the minimum-entropy principle constrained by	 the zeroth full moment and the two half moments of first-order. In this case, the mixed minimum-entropy (MME) ansatz is \cite{Frank07}
\begin{equation}
	\psi_{MME}=\begin{cases} e^{\alpha +\beta_+ \mu}, & \mu\in [0,1], \\ e^{\alpha +\beta_- \mu}, & \mu\in [-1,0]. \end{cases}
\end{equation}
A complete hierarchy of mixed minimum entropy methods of arbitrary order (see Section \ref{sec:MMn}) follows, so we call this model MM${}_1$. Here, one could also consider a linear closure, called the mixed moment MP${}_n$ closure, but this has the drawback that it allows for negative energies and does not adapt to the correct speed of propagation, just as with the full moment P${}_n$ model. 

%We note that, for mixed moment approximations  moment realizability is an issue which  has not been discussed in the previous work. It  will be considered in detail in the next section for a general number of moments.

Again, the system cannot be closed analytically, but the second moments can be written in the form
\begin{align*}
\psi_\pm^{\left(2\right)} = \chi_\pm\left(\frac{\psi_-^{(1)}}{\psi^{(0)}},\frac{\psi_+^{(1)}}{\psiv{0}}\right)\psiv{0}
\end{align*}
%\begin{align*}
% \partial_t \psi^{(0)} + \partial_x (\psi_+^{(1)}+\psi_-^{(1)})  + \sigma_a \psi^{(0)}  &= Q^{(0)}\\
%	 \partial_t \psi_+^{(1)} + \partial_x \left(\chi_+\left(\frac{\psi_+^{(1)}}{\psi^{(0)}},\frac{\psi_-^{(1)}}{\psi^{(0)}}\right) \psi^{(0)}\right) + \sigma_a \psi_+^{(1)}
%	&=\frac{T}{2}(2\psi^{(0)}-5\psi_+^{(1)}+3\psi_-^{(1)}) + Q_+^{(1)}\\
% \partial_t \psi_-^{(1)} + \partial_x\left(\chi_-\left(\frac{\psi_-^{(1)}}{\psi^{(0)}},\frac{\psi_+^{(1)}}{\psi^{(0)}}\right)\psi^{(0)}\right) + \sigma_a \psi_-^{(1)}
%	&=\frac{T}{2}(-2\psi^{(0)}+3\psi_+^{(1)}-5\psi_-^{(1)})+ Q_-^{(1)}.
%\end{align*}
where the Eddington factors  $\chi_+,\chi_-$ have to be determined numerically as for the full moment ansatz.

We note that in \cite{Frank07} the mixed-moment MP${}_n$ or minimum-entropy closures have only been discussed for the lowest order case.
 
In the following we will study mixed moment problems of arbitrary order.
%Section \ref{mixed} contains a proof for the general case.
%On the other hand, the  development of higher order mixed moment closures of the balance equations is investigated in section \ref{closure}. The  resulting set of equations are numerically compared with each other and the Fokker-Planck equation in section \ref{sec:SIM}.

%%%%%%%%%%%%%%%%%%%%%%%%%%
\section{Realizability of mixed moment problems}
\label{mixed}
%%%%%%%%%%%%%%%%%%%%%%%%%%
%
In this section we state necessary and sufficient conditions for the mixed moments to be consistent with a positive distribution function. 
\begin{definition}
Given a vector of real numbers $\left(\psiv{0},\psip{1},\ldots,\psip{n},\psim{1},\ldots,\psim{n}\right)$, the \emph{truncated Hausdorff mixed-moment problem} on $[-1, 1]$ entails finding a nonnegative distribution function $\psi$ such that
\begin{subequations}
\label{eq:MRp}
\begin{align}
\int_{-1}^1 \psi(\mu)~ d\mu &= \psiv{0}\\
\intp{\mu^j \psi(\mu)} &= \psip{j}~~~~0\leq j \leq n\\
\intm{\mu^j \psi(\mu)} &= \psim{j}~~~~0\leq j \leq n
\end{align}
\end{subequations}
Furthermore we denote the realizable domain of vectors for which a solution to this problem exists by $\MR{n} \subset \mathbb{R}^{2n + 1}$.
\end{definition}

To adapt \thmref{thm:FullMomentRealizability} to our mixed-moment problem, we use a basic fact and an elementary lemma.

\begin{fact}
The mixed-moment data $\gamma$ is realizable if and only if there exist $\psip{0}$ and $\psim{0}$ such that $\psiv{0} = \psip{0}$ + $\psim{0}$, and the moments $(\psip{0}, \psip{1}, \ldots , \psip{n})$ and $(\psim{0}, \psim{1}, \ldots , \psim{n})$ are realizable under the Hausdorff conditions for $[0, 1]$ and $[-1, 0]$ respectively.
\end{fact}

The lemma we use appears in a slightly different form as Lemma 2.3 in \cite{CurFial91}. We can use it to show that, according to the realizability conditions in \thmref{thm:FullMomentRealizability}, the moments of order $1, \ldots, n$ for each half interval $[0, 1]$ and $[-1, 0]$ define lower bounds for the quantities $\psip{0}$ and $\psim{0}$ respectively.%
\footnote{
If we consider the more general truncated mixed moment Hausdorff problem over two adjacent intervals $[\mu_-, \mu_0]$ and $[\mu_0, \mu_+]$ (instead of specifically $[-1, 0]$ and $[0, 1]$), depending on the signs of $\mu_-$, $\mu_0$, and $\mu_+$, the conditions of \lemref{lem:pd-ext} can also lead to upper bounds on $\psiv{0}$.  This does not occur when $\mu_0 = 0$, as in our case.
} %
Below we use $\range{M}$ to indicate the linear space spanned by the columns of the matrix $M$.

\begin{lemma}
\label{lem:pd-ext}
Let $A \in \R^{(k + 1) \times (k + 1)}$ be symmetric, $C \in \R^{k \times k}$ be symmetric, $\bb \in \R^k$, and $a \in \R$ so that
$$
A = \left(\begin{array}{cc} a & \bb^T \\ \bb & C \end{array}\right).
$$
\begin{itemize}
 \item[(i)] If $A \ge 0$, then $C \ge 0$, $\bb \in \range{C}$, and $a \ge \bw^T C \bw$, where $\bb = C \bw$.
 \item[(ii)] If $C \ge 0$ and $\bb \in \range{C}$, then $A \ge 0$ if and only if $a \ge \bw^T C \bw$, where $\bb = C \bw$.
\end{itemize}
\end{lemma}

\begin{remark}
Since $C$ is invertible on its range, a more explicit formula for the bound on $a$ in Lemma \ref{lem:pd-ext} can be written using the pseudo-inverse $C^\dag$ of $C$.  That is, for all $\bw$ such that $\bb = C \bw$, we have $\bw^T C \bw = \bb^T C^\dag \bb$.  Thus the bound on $a$ is well-defined even when $C$ is singular.

\end{remark}

Thus the mixed-moment data $\left(\psiv{0},\psip{1},\ldots,\psip{n},\psim{1},\ldots,\psim{n}\right)$ is realizable only if $\psiv{0}$ is large enough that it can be split up into $\psip{0}$ and $\psim{0}$ which each satisfy the lower-bounds
imposed by $\left( \psip{1}, \psip{1}, \ldots , \psip{n} \right)$ and $\left( \psim{1}, \psim{2}, \ldots , \psim{n} \right)$ respectively through Lemma \ref{lem:pd-ext} and the appropriate Hausdorff conditions.  This gives the following result.

\begin{theorem}
\label{thm:MMRealizability}
Let $B_{\pm}(k)$ and $C_{\pm}(k)$ be defined as in \thmref{thm:FullMomentRealizability} but with $\psiv{i}$ replaced by $\psipm{i}$ (respectively) for $i \ge 1$.  We also let $D(k)_{\pm} := \left( \psipm{i + j + 1} \right)_{i, j = 1}^k$ and define the mixed-moment matrices $\tilde A_{\pm}(k)$ using both the full moment $\psiv{0}$ and the partial moments by
$$
\left( \tilde A_{\pm}(k) \right)_{ij} = \begin{cases} \psiv{0} &\text{if } i = j = 0 \\
                               \psipm{i + j} &\text{otherwise} \end{cases}
 \quad i, j \in \{ 0, 1, \ldots , k \}.
$$

Then the mixed-moment data $\gamma$ is realizable if and only if
\begin{itemize}
 \item[(i)] when $N = 2k$ is even,
 \begin{subequations}
 \label{eq:mixed-even}
 \begin{align}
 & \pm B_{\pm}(k - 1) - C_{\pm}(k) \ge 0, \\
 & \tilde A_{\pm}(k) \ge 0, \label{eq:mixed-even-2} \\
 & \psiv{0} \ge \bb_+^T C_+(k)^\dag \bb_+ + \bb_-^T C_-(k)^\dag \bb_-, \label{eq:mixed-even-glue}
 \end{align}
 \end{subequations}
 where $\bb_{\pm} := \left(\psipm{1}, \ldots , \psipm{k}\right)^T$ are simply the first columns of $A_{\pm}(k)$---but omitting the top element---respectively, and $C_{\pm}(k)^\dag$ represent the pseudo-inverses of $C_{\pm}(k)$ respectively;
 \item[(ii)] when $N = 2k + 1$ is odd,
 \begin{subequations}
 \label{eq:mixed-odd}
 \begin{align}
 & \pm B_{\pm}(k) \ge 0, \\
 & \tilde A_{\pm}(k) \mp B_{\pm}(k) \ge 0,  \label{eq:mixed-odd-2}\\
 & \psiv{0} \ge \psip{1}
  + \bb_+^T \left( C_+(k) - D_+(k) \right)^\dag \bb_+ \nonumber \\
  & \qquad - \psim{1}
  + \bb_-^T \left( C_-(k) + D_-(k) \right)^\dag \bb_-, \label{eq:mixed-odd-glue}
 \end{align}
 \end{subequations}
 where here $\bb_{\pm} := \left( \psipm{1} - \psipm{2}, \ldots , \psipm{k} - \psipm{k + 1} \right)^T$ are simply the first \\ columns of $A_{\pm}(k) \mp B_{\pm}(k)$---but omitting the top element---respectively, and $\left( C_{\pm}(k) \mp D_{\pm}(k) \right)^\dag$ represent the pseudo-inverses of $C_{\pm}(k) \mp D_{\pm}(k)$ respectively.
\end{itemize}

\end{theorem}

\begin{proof}
We first prove the necessity and sufficiency of the following conditions, which follow more directly from \lemref{lem:pd-ext}:
\begin{itemize}
 \item[(i)] when $N = 2k$ is even,
 \begin{subequations}
 \label{eq:mixed-even-r}
 \begin{align}
 & \pm B_{\pm}(k - 1) - C_{\pm}(k) \ge 0, \label{eq:mixed-even-1r} \\
 & C_{\pm}(k) \ge 0, \label{eq:mixed-even-2r} \\
 & \bb_{\pm} \in \range{C_{\pm}(k)}, \label{eq:mixed-even-range} \\
 & \psiv{0} \ge \bb_+^T C_+(k)^\dag \bb_+ + \bb_-^T C_-(k)^\dag \bb_-;
  \label{eq:mixed-even-glue-r}
 \end{align}
 \end{subequations}
 
 \item[(ii)] when $N = 2k + 1$ is odd,
 \begin{subequations}
 \begin{align}
 & \pm B_{\pm}(k) \ge 0, \\
 & C_{\pm}(k) \mp D_{\pm}(k) \ge 0, \\
 & \bb_{\pm} \in \range{C_{\pm}(k) \mp D_{\pm}(k)},  \label{eq:mixed-odd-range} \\
 & \psiv{0} \ge \psip{1}
  + \bb_+^T \left( C_+(k) - D_+(k) \right)^\dag \bb_+ \nonumber \\
  & \qquad - \psim{1}
  + \bb_-^T \left( C_-(k) + D_-(k) \right)^\dag \bb_-. 
 \end{align}
 \end{subequations}
\end{itemize}

We first quickly note that both sets of conditions have the following form:  they first include the Hausdorff conditions which do not involve zeroth-order moments, followed by the conditions of Lemma \ref{lem:pd-ext}, from which the final condition for each side is summed to give a total lower bound on $\psiv{0}$.  We prove just the even case, as the proof in the odd case is analogous.
  
For necessity of \eqref{eq:mixed-even-r}, assume that we have a density $\psi$ which represents $\gamma$.  Then let $\psip{0} := \int_0^1 \psi(\mu) d\mu$ and $\psim{0} := \int_{-1}^0 \psi(\mu) d\mu$, so that clearly $\psiv{0} = \psip{0} + \psim{0}$.  Then \eqref{eq:mixed-even-1r} follows immediately from \eqref{eq:haus-even-2}.  Conditions \eqref{eq:mixed-even-2r}-\eqref{eq:mixed-even-glue-r} follow from \eqref{eq:haus-even-1} and Lemma \ref{lem:pd-ext}, where $a = \psipm{0}$, and \eqref{eq:mixed-even-glue-r} is obtained simply by summing the final conditions Lemma \ref{lem:pd-ext} for each half interval.

For sufficiency \eqref{eq:mixed-even-r}, we choose $\psip{0} = \bb_+^T C_+(k - 1)^\dag \bb_+$ and $\psim{0} = \psiv{0} - \bb_+^T C_+(k - 1)^\dag \bb_+ \ge \bb_-^T C_-(k)^\dag \bb_-$.  With these values and conditions \eqref{eq:mixed-even-2r}-\eqref{eq:mixed-even-glue-r}, we can use Lemma \ref{lem:pd-ext} to construct positive-definite matrices $A_{\pm}(k)$.  These matrices together with condition \eqref{eq:mixed-even-1r} are the Hausdorff conditions.  Therefore we have densities defined on both $[0, 1]$ and $[-1, 0]$ which we can concatenate to represent $\gamma$.

Finally, it is not hard to see that, again via \lemref{lem:pd-ext}, \eqref{eq:mixed-even-2r}-\eqref{eq:mixed-even-glue-r} are equivalent to \eqref{eq:mixed-even-2}-\eqref{eq:mixed-even-glue}
\end{proof}

\begin{example}
We examine the coupling conditions \eqref{eq:mixed-even-glue} and \eqref{eq:mixed-odd-glue} for $n=2$ and $n=3$ more explicitly:

When $n = 2$, we have $k = 1$, we have $C_{\pm}(1) = \psipm{2}$, whose pseudo-inverses are $0$ when $\psipm{2} = 0$.  Therefore
$$
\bb_+^T C_+(1)^\dag \bb_+
 = \begin{cases} \frac{ \left( \psip{1} \right)^2}{\psip{2}} &\mbox{if } \psip{2} \ne 0, \\
 0 & \mbox{otherwise}.
 \end{cases} 
$$
The range conditions \eqref{eq:mixed-even-range} are only nontrivial in the singular cases, which for $n = 2$ are when $\psip{2} = 0$ or $\psim{2} = 0$.  In these cases the range conditions require $\psipm{1} = 0$ respectively (which are consistent with the fact that here we must have %
%$\support{\left. \psi \right|_{[0, 1]}} \subseteq \{ 0 \}$ or $\support{\left. \psi \right|_{[-1, 0]}} \subseteq \{ 0 \}$
$\support{\psi} \cap [0, 1] \subseteq \{ 0 \}$ or $\support{\psi} \cap [-1, 0] \subseteq \{ 0 \}$ %
respectively, for any representing $\psi$). But in the nonsingular case, condition \eqref{eq:mixed-even-glue} reads
$$
\psiv{0} \ge \frac{ \left( \psip{1} \right)^2}{\psip{2}} + \frac{ \left( \psim{1} \right)^2}{\psim{2}}.
$$

When $n = 3$, we again have $k = 1$, and
$$
\bb_+^T \left( C_+(1) - D_+(1) \right)^\dag \bb_+
 = \begin{cases} \frac{ \left( \psip{1} - \psip{2} \right)^2}{\psip{3} - \psip{2}} &\mbox{if } \psip{3} \ne \psip{2}, \\
 0 & \mbox{otherwise}.
 \end{cases} 
$$
In the singular case the range conditions \eqref{eq:mixed-odd-range} impose $\psipm{1} = \pm \psipm{2}$ (which are consistent with the realizability conditions and the fact that here we must have %
%$\support{\left. \psi \right|_{[0, 1]}} \subseteq \{ 0, 1 \}$ or $\support{\left. \psi \right|_{[-1, 0]}} \subseteq \{ -1, 0 \}$
$\support{\psi} \cap [0, 1] \subseteq \{ 0, 1 \}$ or $\support{\psi} \cap [-1, 0] \subseteq \{ -1, 0 \}$ %
respectively for any representing $\psi$).  Thus in the nonsingular case, condition \eqref{eq:mixed-odd-glue} reads
$$
\psiv{0} \ge \psip{1} + \frac{ \left( \psip{1} - \psip{2} \right)^2}{\psip{3} - \psip{2}}
  - \psim{1} + \frac{ \left( \psim{1} + \psim{2} \right)^2}{\psim{3} + \psim{2}}.
$$

\end{example}
\begin{remark}
\label{ex:GenFunctions}
As for full moments (see Remark \ref{rem:GeneratingFunction}) atoms $\mu_{i\pm}$ for the distribution function are the roots of the generating function with $\gamma_{i} = \left(\pm 1\right)^i\Phipm{i}$ respectively. The densities $\rho_{i\pm}$ can be calculated afterwards from the corresponding Vandermonde system. The generating functions $g_\gamma(k) = \mu^k-\left(\sum\limits_{i=0}^{k-1}\varphi_i\mu^i\right) $ for $n\leq 6$ and $k=\ceil{\frac{n}{2}}$ are given by $\varphi$ in the next table:\\
\begin{tabular}{l|cc}
$k$&$2k-1$&$2k$\\
\hline
$1$&$\left(1\right)$&$\left(\frac{\gamma_2}{\gamma_1}\right)$\\
$2$&$\left(-\frac{\gamma_2 - \gamma_3}{\gamma_1 - \gamma_2},\frac{\gamma_1 - \gamma_3}{\gamma_1 - \gamma_2}\right)^T$&$\left(-\frac{- \gamma_3^2 + \gamma_2\gamma_4}{- \gamma_2^2 + \gamma_1\gamma_3},\frac{\gamma_1\gamma_4 - \gamma_2\gamma_3}{- \gamma_2^2 + \gamma_1\gamma_3}\right)^T$\\
$3$&$\begin{pmatrix}\frac{\gamma_3\gamma_4 + \gamma_3\gamma_5 + \gamma_2(\gamma_4 - \gamma_5) - \gamma_3^2 - \gamma_4^2}{\gamma_2\gamma_3 + \gamma_2\gamma_4 + \gamma_1(\gamma_3 - \gamma_4) - \gamma_2^2 - \gamma_3^2}\\ \frac{\gamma_1\gamma_4 - \gamma_2\gamma_3 - \gamma_1\gamma_5 + \gamma_2\gamma_4 + \gamma_3\gamma_5 - \gamma_4^2}{\gamma_1\gamma_4 - \gamma_2\gamma_4 + \gamma_2^2 + \gamma_3^2 - \gamma_3(\gamma_1 + \gamma_2)}\\     -\frac{\gamma_2\gamma_4 - \gamma_1\gamma_5 + \gamma_2\gamma_5 + \gamma_3(\gamma_1 - \gamma_4) - \gamma_2^2}{\gamma_1\gamma_4 - \gamma_2\gamma_4 + \gamma_2^2 + \gamma_3^2 - \gamma_3(\gamma_1 + \gamma_2)}\end{pmatrix}$&$\begin{pmatrix}
\frac{\gamma_2(- \gamma_5^2 + \gamma_4\gamma_6) - \gamma_3^2\gamma_6 - \gamma_4^3 + 2\gamma_3\gamma_4\gamma_5}{\gamma_1(- \gamma_4^2 + \gamma_3\gamma_5) - \gamma_2^2\gamma_5 - \gamma_3^3 + 2\gamma_2\gamma_3\gamma_4}\\\frac{\gamma_3^2\gamma_5 + (- \gamma_4^2 - \gamma_2\gamma_6)\gamma_3 - \gamma_1\gamma_5^2 + \gamma_2\gamma_4\gamma_5 + \gamma_1\gamma_4\gamma_6}{\gamma_5\gamma_2^2 - 2\gamma_2\gamma_3\gamma_4 + \gamma_3^3 - \gamma_1\gamma_5\gamma_3 + \gamma_1\gamma_4^2}\\  \frac{\gamma_3^2\gamma_4 - \gamma_2\gamma_4^2 - \gamma_3(\gamma_1\gamma_6 + \gamma_2\gamma_5) + \gamma_2^2\gamma_6 + \gamma_1\gamma_4\gamma_5}{\gamma_5\gamma_2^2 - 2\gamma_2\gamma_3\gamma_4 + \gamma_3^3 - \gamma_1\gamma_5\gamma_3 + \gamma_1\gamma_4^2}\end{pmatrix}$
\end{tabular}\\
\vskip 0.1cm
\end{remark}

\section{Mixed Moment Closures}
\label{closure}
In this section we derive the mixed moment MP${}_n$ closure for arbitrary order as well as the minimum entropy closure. A special discrete closure which we call the Kershaw closure is given up to second order.

\begin{remark}
We want to emphasize that all models derived here are hyperbolic (strictly inside the realizability domain). Although we do not prove it here, it is easy to see, using similar arguments as for full moments \cite{Lev96}, that mixed minimum entropy models and mixed MP${}_n$ models fulfill this property. Additionally all eigenvalues are bounded in absolute value by one.

For (mixed moment) Kershaw closures we know of no general proof for hyperbolicity but it is easy to check for each model separately.
\end{remark}

\subsection{MP${}_n$}
The mixed MP${}_n$ closure consists of a basis function which is (in every) half-space a polynomial in the angular variable $\mu$. We additionally demand the distribution function to be continuous in $\mu=0$. This results in the following ansatz:
\begin{align*}
\psi_{\text{MP}_n} = \begin{cases}
\alpha +{\sum\limits_{i=1}^{n}\beta_+^{\left(i\right)}\mu^i} & \text{for } \mu\in [0,1]\\
\alpha +{\sum\limits_{i=1}^{n}\beta_-^{\left(i\right)}\mu^i} & \text{for } \mu\in  [-1,0]
\end{cases}
\end{align*}
We use monomials here because spherical harmonics (that is, Legendre polynomials) lose their orthogonality on the half-spaces and are therefore not superior to the usual monomial basis.
% However the distribution function $\psi_{P_n}$ is not necessarily positive and will therefore not fulfil the conditions derived above.
% 
\subsection{Mixed Minimum Entropy}
\label{sec:MMn}
%In \cite{Frank07}  this model has been derived for first order to overcome the problems of the full moment minimum entropy (zero netflux and therefore unphysical shocks). As for the full moment minimum entropy one asks in the first order case for the solution which minimizes the  entropy
%\begin{align*}
%&\min\limits_{\psi}  \intv{-\psi\log\left(\psi\right)} \\
%&\text{subject to: } \left(\intv{\psi},\intp{\mu\psi},\intm{\mu\psi}\right)=\left(\psiv{0},\psip{1},\psim{1}\right)
%\end{align*}
%which has to first order the solution
%\begin{align*}
%\psi_{\text{MME}} = \begin{cases}
%\alpha e^{\bp\mu} & \text{for } \mu\in [0,1]\\
%\alpha e^{\bm\mu} & \text{for } \mu\in [-1,0]
%\end{cases}
%\end{align*}
%where $\alpha,\bp,\bm$ are given implicity by the interpolation constraint.
The general MM${}_n$ is obtained by the same entropy ansatz as for the original MME/MM${}_1$. The entropy minimizer is given by
\begin{align*}
\psi_{\text{MM}_n} = \begin{cases}
\exp\left(\alpha + \sum\limits_{i=1}^{n}\beta_+^{\left(i\right)}\mu^i \right) & \text{for } \mu\in [0,1]\\
\exp\left(\alpha +\sum\limits_{i=1}^{n}\beta_-^{\left(i\right)}\mu^i \right) & \text{for } \mu\in  [-1,0]
\end{cases}
\end{align*}
As proposed in \cite{Frank07} a tabulation is used for $n=1$ to avoid a nonlinear solution technique as Newton's method. However this seems not appropriate for larger $n$ since the tableau has dimension $2n$. Here a robust algorithm has to be applied, especially at the boundary of the realizability domain, see e.g. \cite{AllHau12} for more details on this topic for the original minimum entropy models. Additionally, this task is even more challenging for $n>2$ because in this case the integrals cannot be solved analytically. Since this is not the main topic of this paper we will only give numerical examples for $n=1$.

\subsection{Kershaw closures}
To avoid the nonlinear inversion process of the minimum entropy approach we want to construct a closure which can be computed analytically. The key idea for the mixed moments is identical to the idea for the full moments explained in Section \ref{sec:kershaw} for an introduction. Although arbitrary high orders are in principle available due to Theorem \ref{thm:MMRealizability} we only present Kershaw closures up to order $2$. As in Section \ref{sec:kershaw} convex combinations of "upper" and "lower" higher order realizability distributions  $\psi_{up},\psi_{low}$ are needed. However, finding suitable candidates which are symmetric in terms of the two halfspaces which reproduces moments on the coupling conditions \eqref{eq:mixed-even-glue} and \eqref{eq:mixed-odd-glue} is a non-trivial task and may be possible only due to intensive symbolic calculations.
\subsubsection{MK${}_1$}
We want to construct a distribution function $\psi_{MK_1}$ which is on the one hand realizable and on the other hand interpolates the equilibrium point \\$\left(\Phip{1},\Phim{1},\Phip{2},\Phim{2}\right) = \left(\frac{1}{4},-\frac{1}{4},\frac{1}{6},\frac{1}{6}\right)$.
 One idea would be to do a convex combination between the isotropic state $\psi_{const}$ and the free streaming limit calculated in Theorem \ref{thm:MMRealizability}:
 \begin{align}
 \label{eq:MK1failure}
 \psip{2} = \alpha\intp{ \mu^2 \psi_{const}} + (1-\alpha)\intp{ \mu^2 \psi_{up}}
 \end{align}
 with 
 \begin{align*}
 \psi_{up}  &= \psiv{0}\left(\Phip{1}\delta\left(1-\mu\right) -\Phim{1}\delta\left(1+\mu\right) + \left(1-\Phip{1}+\Phim{1}\right)\delta\left(\mu\right)\right)
 \end{align*}
 and $\alpha\left( \Phip{1},\Phim{1} \right)$ such that $\alpha\left(\frac14,-\frac14\right) = 1$. This works well for full moments (see e.g. \cite{Monreal}) but fails for mixed moments. In the isotropic limit $\psi\to \psi_{const}$ there is no unique description of this state in terms of available moments, that is, the problem is underdetermined. Choosing e.g. $\psi_{const} = \frac{1}{2}\psiv{0} = 2\psip{1} = -2\psim{1}$ gives different closures in \eqref{eq:MK1failure}. Correspondingly, different choices of $\psi_{const}$ result in systems with different eigenvalues for the Jacobian of the flux function. Choosing naively $\alpha = 1$ and $\psipm{2} = \cfrac{\psiv{0}}{6}$ in the first example and $\psipm{2} = \pm\cfrac{2\psipm{1}}{6}$ in the second example gives system matrices
 \begin{align*}
 M_1 = \begin{pmatrix}
 0 &1 &1\\0 &\frac{2}{3} &0\\0&0&-\frac{2}{3}\end{pmatrix}~~~~~~~~M_2 = \begin{pmatrix}
 0 &1 &1\\\frac{1}{6} &0 &0\\\frac{1}{6}&0&0\end{pmatrix}.
 \end{align*}
 which have the following set of eigenvalues: $\lambda_1 = \left(-\frac{2}{3},0,\frac{2}{3}\right)$ and $\lambda_2 = \left(-\sqrt{\frac{1}{3}},0,\sqrt{\frac{1}{3}}\right)$.
  \\

To overcome this problem we use again, as in Section \ref{sec:kershaw} for the K${}_1$ model, a convex combination of the upper and lower realizability boundary for $n=2$. 
\begin{align*}
\psi_{low}  &= \psiv{0}\left(\frac{\Phip{1}}{\Phip{1}-\Phim{1}}\delta\left(\mu-\left(\Phip{1}-\Phim{1}\right)\right) -\frac{\Phim{1}}{\Phip{1}-\Phim{1}}\delta\left(\mu+\left(\Phip{1}-\Phim{1}\right)\right)\right)
\end{align*}
A convex combination of those implies second moments
\begin{align*}
\Phipm{2} &= \pm\alpha\Phipm{1} \pm \left(1-\alpha\right) \Phipm{1} \left(\Phip{1}-\Phim{1}\right)
\end{align*}
Since the equilibrium point should be interpolated we can conclude $\alpha = \frac{1}{3}$. However, every non-constant solution $0\leq \alpha\left(\Phip{1},\Phim{1}\right)\leq 1$ with $\alpha\left(\frac14,-\frac14\right) = \frac16$ would be appropriate since at the lower order boundaries the two distributions coincide.\\

In \figref{fig:fluxDiff} these lower (\ref{fig:MK1LowerBoundary}) and upper (\ref{fig:MK1UpperBoundary}) boundary representations of $\Phip{2}$ are shown, as well as the normalized positive second moment $\Phip{2}$ of the MK${}_1$ model which is then compared with the one calculated for the mixed minimum entropy solution. Note that for $\Phim{2}$ the result is just mirrored along the bisecting line of this triangle.

\InsertFig{\subfloat[Lower boundary representation]{
         \includegraphics[width = 0.45\textwidth]{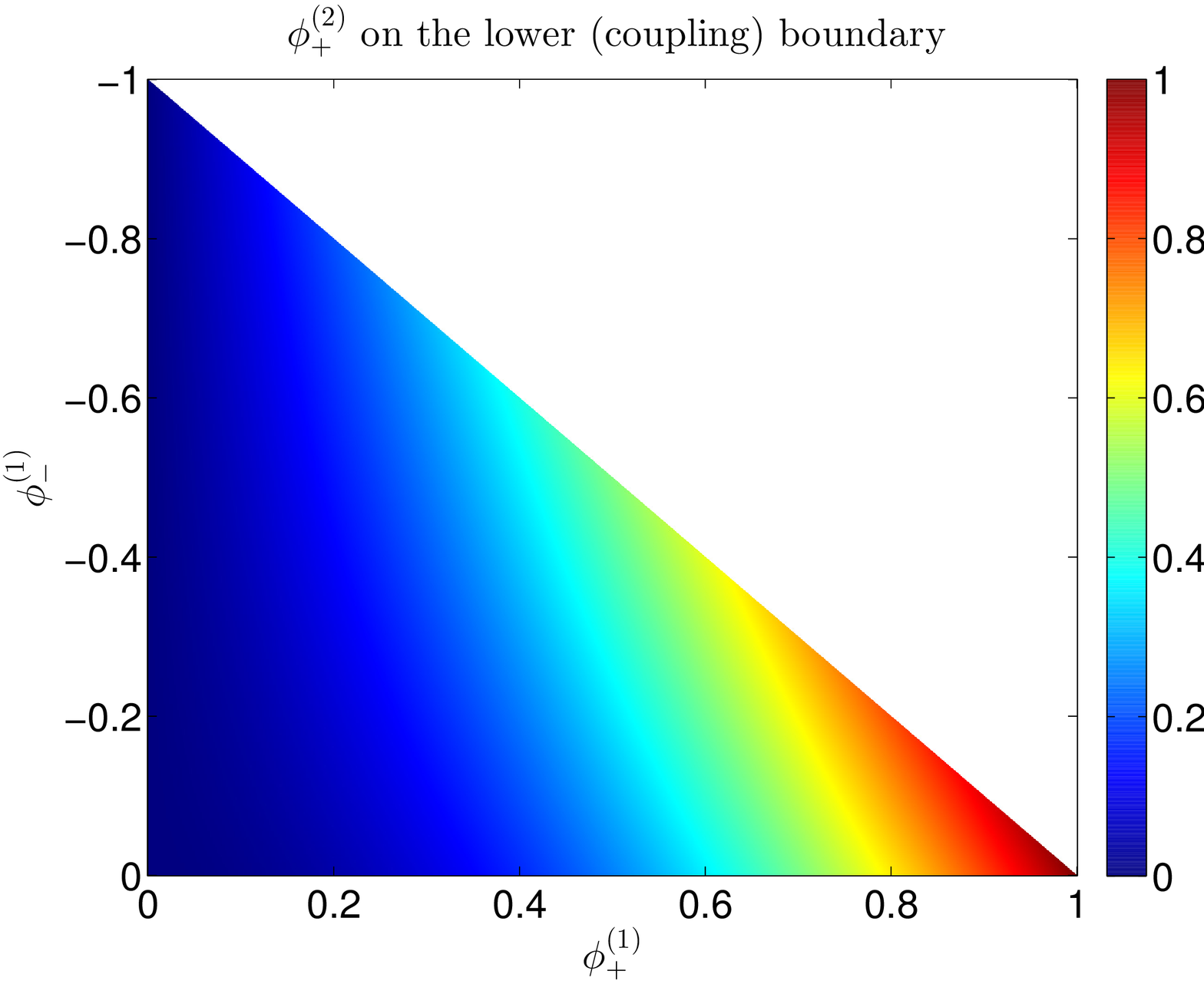}
         \label{fig:MK1LowerBoundary}
         }
\subfloat[Upper boundary representation]{
         \includegraphics[width = 0.45\textwidth]{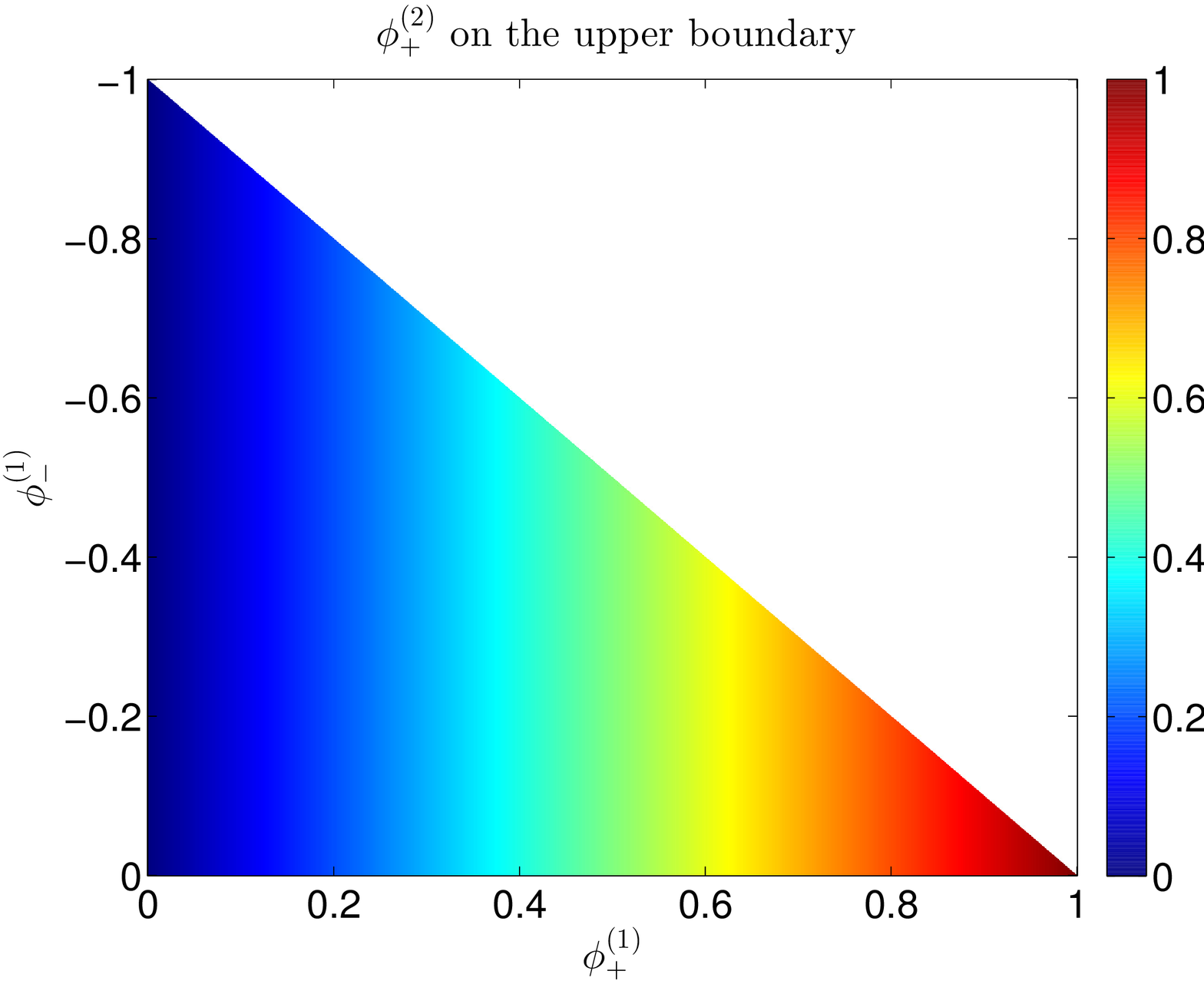}
         \label{fig:MK1UpperBoundary}}
         
\subfloat[MK${}_1$-Flux]{
         \includegraphics[width = 0.45\textwidth]{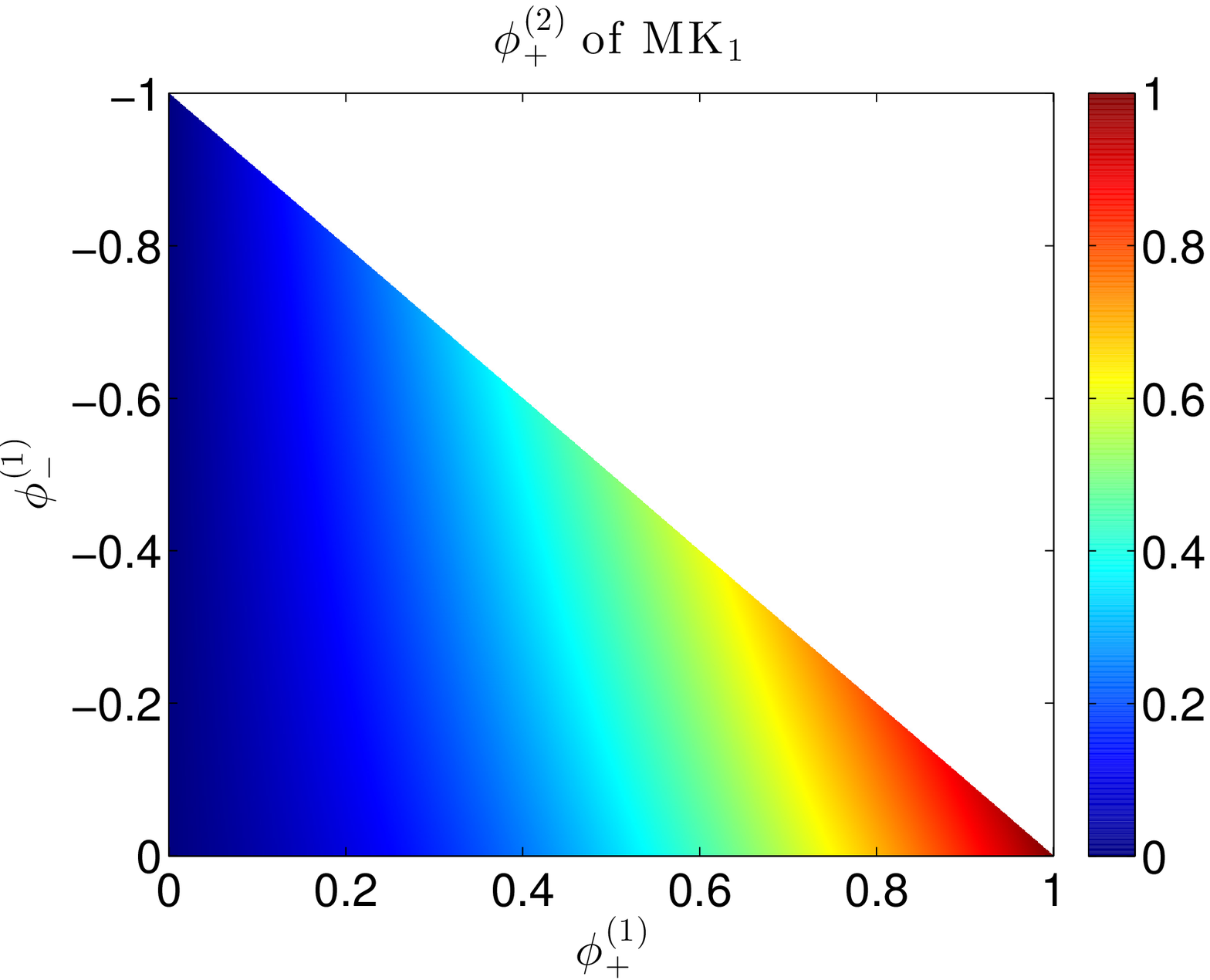}
         \label{fig:MK1Flux}
         }
\subfloat[Difference of MM${}_1$ and MK${}_1$]{
         \includegraphics[width = 0.45\textwidth]{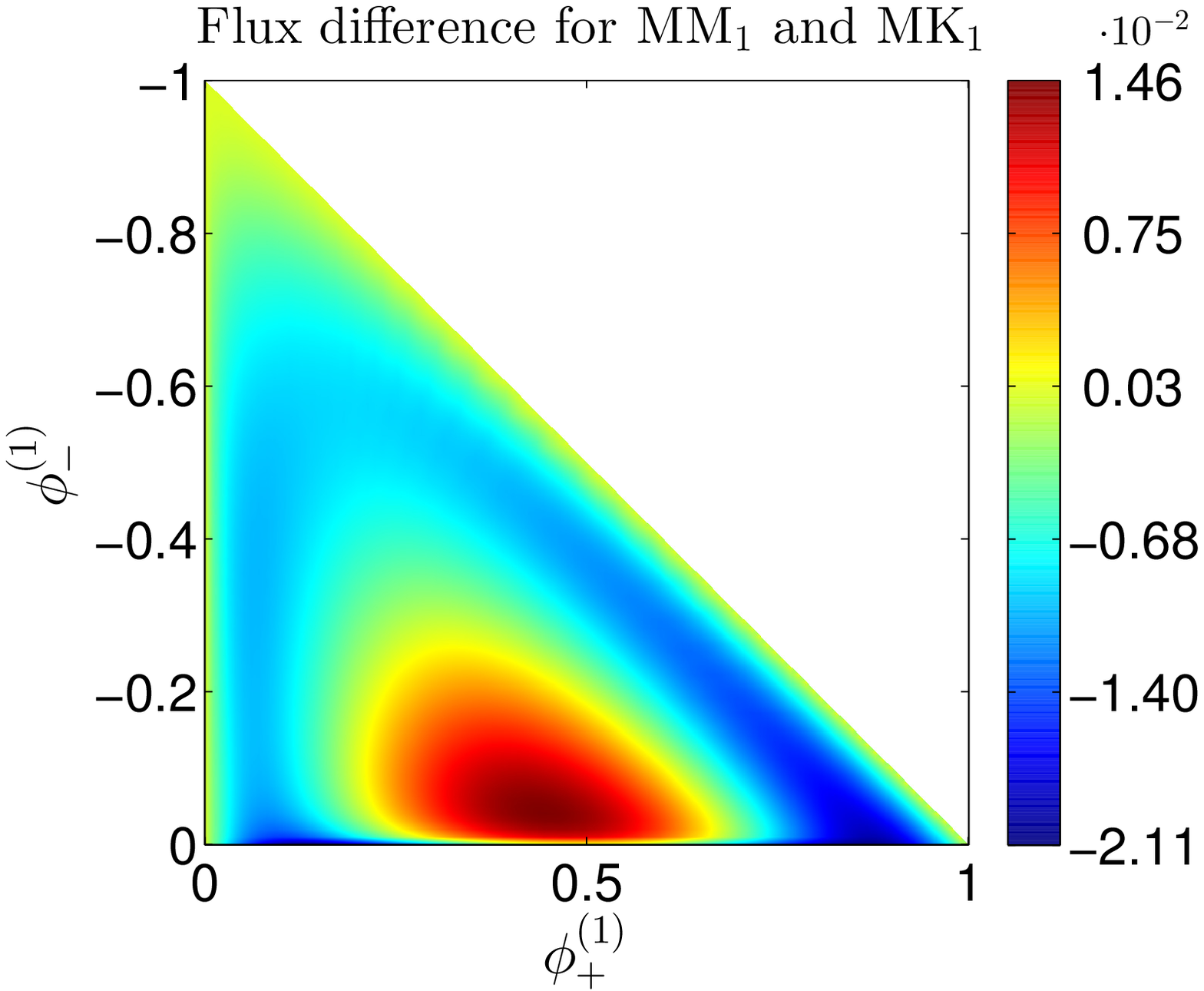}
         \label{fig:FluxDifference}}
         }{$\Phip{2} = \Phip{1} - \Phim{1}$ (\ref{fig:MK1LowerBoundary}), $\Phip{2} = \Phip{1}$ (\ref{fig:MK1UpperBoundary}), $\Phip{2}_{\text{MK}_1}$ and $\Phip{2}_{\text{MM${}_1$}}-\Phip{2}_{\text{MK}_1}$.}{\label{fig:fluxDiff}} 
\subsubsection{MK${}_2$}
As before we are looking for two distribution functions which lie on the lower and upper realizability boundary of order $3$, respectively. The corresponding realizability conditions are
\begin{subequations}
\begin{gather}
\frac{\Phipm{2}^2}{\pm\Phipm{1}} \leq \pm \Phipm{3} \leq \Phipm{2}-\cfrac{\left(\pm\Phipm{1}-\Phipm{2}\right)^2}{1\mp\Phipm{1}}\\
\Phip{1}+\cfrac{\left(\Phip{1}-\Phip{2}\right)^2}{\Phip{2}-\Phip{3}}-\Phim{1}+\cfrac{\left(-\Phim{1}-\Phim{2}\right)^2}{\Phim{2}+\Phim{3}}\leq 1
\end{gather}
\end{subequations}

We observe that
\small\begin{align*}
\psi_{low} = \psiv{0}\left(\frac{\Phip{1}^2}{\Phip{2}}\delta\left(\mu-\frac{\Phip{2}}{\Phip{1}}\right)+\frac{\Phim{1}^2}{\Phip{2}}\delta\left(\mu-\frac{\Phim{2}}{\Phip{1}}\right)+\left(1-\frac{\Phip{1}^2}{\Phip{2}}-\frac{\Phim{1}^2}{\Phim{2}}\right)\delta\left(\mu\right)\right)
\end{align*}\normalsize
reproduces
\begin{align*}
{\psiv{3}_\pm}_{low} = \int{\mu^3\psi~d\mu} = \psiv{0}\frac{\Phipm{1}^2}{\Phipm{2}}\left(\frac{\Phipm{2}}{\Phipm{1}}\right)^3 = \psiv{0}\frac{\Phipm{2}^2}{\Phipm{1}}
\end{align*}
Since the formula for $\psi_{up}$ is very complicated we only state the generated normalized third-order moments which satisfy the upper boundary/coupling conditions symmetrically in terms of the half-intervals:
\begin{align*}
\Phip{3}_{up} &=\Phip{2} - \frac{{\left(\Phip{1} - \Phip{2}\right)}^2}{\Phim{1} - \Phip{1} - \frac{{\Phim{1}}^2 + \Phim{2}\, \Phim{1}}{\Phim{2} + \frac{\Phip{2}\, {\Phim{1}}^2 + \Phim{2}\, {\Phip{1}}^2 - \Phim{2}\, \Phip{2}}{\Phip{2}\, \left(\Phim{1} - \Phip{1} + 1\right)}} + 1}\\
\Phim{3}_{up} &= \frac{{\left(\Phim{1} + \Phim{2}\right)}^2}{\Phim{1} - \Phip{1} + \frac{\Phip{1}\, \Phip{2} - {\Phip{1}}^2}{\Phip{2} + \frac{\Phip{2}\, {\Phim{1}}^2 + \Phim{2}\, {\Phip{1}}^2 - \Phim{2}\, \Phip{2}}{\Phim{2}\, \left(\Phim{1} - \Phip{1} + 1\right)}} + 1} - \Phim{2}.
\end{align*}
These moments satisfy the third order coupling condition with equality.\\
In the equilibrium point we have $\Phipm{3} = \pm\frac{1}{8}$. Therefore
\begin{align*}
\Phip{3} &= \alpha\frac{\Phip{2}^2}{\Phip{1}} + \left(1-\alpha\right)\left(\Phip{2} - \frac{{\left(\Phip{1} - \Phip{2}\right)}^2}{\Phim{1} - \Phip{1} - \frac{{\Phim{1}}^2 + \Phim{2}\, \Phim{1}}{\Phim{2} + \frac{\Phip{2}\, {\Phim{1}}^2 + \Phim{2}\, {\Phip{1}}^2 - \Phim{2}\, \Phip{2}}{\Phip{2}\, \left(\Phim{1} - \Phip{1} + 1\right)}} + 1}\right)\\
\Phim{3} &= \alpha\frac{\Phim{2}^2}{\Phim{1}}+\left(1-\alpha\right)\left(\frac{{\left(\Phim{1} + \Phim{2}\right)}^2}{\Phim{1} - \Phip{1} + \frac{\Phip{1}\, \Phip{2} - {\Phip{1}}^2}{\Phip{2} + \frac{\Phip{2}\, {\Phim{1}}^2 + \Phim{2}\, {\Phip{1}}^2 - \Phim{2}\, \Phip{2}}{\Phim{2}\, \left(\Phim{1} - \Phip{1} + 1\right)}} + 1} - \Phim{2}\right)
\end{align*}
with $\alpha = \frac{1}{2}$.

\section{Numerical results}
\label{sec:SIM}

%To compare  the  models derived in the last section we will use the following test equation with $x \in \R$, $\mu \in [-1,1]$
%\begin{align}
%\label{eq:FokkerPlanck1D}
%\partial_t\psi + \mu\cdot\partial_x \psi + \sigma_a\psi = \frac{T}{2}\Delta_\mu \psi + Q
%\end{align}
%where $\sigma_a$ is the absorption coefficient, $T$ the scattering coefficient, $Q=Q(x)$ the source and $\Delta_\mu$ given by the one-dimensional projection of the Laplace-Beltrami operator on the sphere:
%\begin{align}
%\Delta_\mu \psi = \frac{\partial}{\partial\mu}\left(\left(1-\mu^2\right) \frac{\partial \psi}{\partial \mu}\right)
%\end{align}
\subsection{Implementation}
As reference we use a standard finite-difference approximation of the Fokker-Planck equation \eqref{eq:FokkerPlanck1D}. Moment approximations are calculated using variants of a standard finite volume scheme \cite{toro2009riemann,leveque1992numerical}
\begin{align*}
u_i^{j+1} = u_i^j - \lambda\left[\frac{1}{h}\int\limits_{t_j}^{t_{j+1}}f\left(u\left(t,x_{i+\frac{1}{2}}\right)\right)~dt - \frac{1}{h}\int\limits_{t_j}^{t_{j+1}}f\left(u\left(t,x_{i-\frac{1}{2}}\right)\right)~dt\right]
\end{align*}
where $\lambda = \frac{\Delta t}{\Delta x}$ and $u_i^j$ denotes the solution at cellcenter $i$ at time $t_j$ of the general hyperbolic system of conservation laws
\begin{align*}
u_t + f\left(u\right)_x = 0
\end{align*}
For first order numerical approximation this can be rewritten in conservative form as
\begin{align*}
u_i^{j+1} = u_i^j - \lambda\left[h(u_i^j,u_{i+1}^j)-h(u_{i-1}^j,u_{i}^j)\right]
\end{align*}
with appropriate numerical fluxes $h(u,v)$. The source terms are approximated consistently using cell-averages of $\sigma_a$, $T$, and $Q$.\\

In all examples we use $n_x = 1000$ points for the spatial discretization while additionally $n_\mu=800$ points in the angular variable are used for the Fokker-Planck solution.\\
\subsubsection{Full moment P${}_n$}
The full moment spherical harmonics are discretized with a Godunov/Upwind scheme:
\begin{align*}
h(u,v) = \frac{1}{2}\left(Au+Av - \abs{A}\left(v-u\right)\right)
\end{align*}
where $A$ is the Jacobian of the flux function $f$ (which is linear for the P${}_n$ equations) and $\abs{A} = A^+ - A^-$ with $A = TDT^{-1}$, $D = \operatorname{diag}(\lambda_1,\ldots,\lambda_n)$
\begin{align*}
\lambda_i^\pm &= \pm\max\left(\pm\lambda_i,0\right)\\
D^\pm &= \operatorname{diag}\left(\lambda_1^\pm,\ldots,\lambda_n^\pm\right)\\
A^\pm &= TD^\pm T^{-1}.
\end{align*}
\subsubsection{Full moment K${}_1$/M${}_1$}
The full moment K${}_1$/M${}_1$ model is solved using an HLL solver (see e.g. \cite{toro2009riemann}):
\begin{align*}
h\left(u_L,u_R\right) = \begin{cases}
F_L & \text{if }  0\leq S_L\\
\cfrac{S_RF_L-S_LF_R+S_LS_R\left(u_R-u_L\right)}{S_R-S_L}& \text{if } S_L\leq 0 \leq S_R\\
F_R & \text{if } 0 \geq S_R. 
\end{cases}
\end{align*}
with $F_{L/R} = f(u_{L/R})$ and $S_{L/R}$ the approximate wave speeds at the left and the right cell-center, respectively. We choose $S_R = 1 = -S_L$ since the eigenvalues for M${}_1$ and K${}_1$ are bounded in absolute value by $1$ \cite{DubFeu99}.

For the minimum entropy model we use 
%\begin{align*}
%N_2 = 1-\frac{2}{Z}\left(\coth Z -\frac{1}{Z}\right)
%\end{align*}
%where $Z$ is implicitly defined by 
%\begin{align*}
%N_1 = \coth Z -\frac{1}{Z}
%\end{align*}
the numerical approximation from \cite{Ducl2010} which uses a rational fit for $\Phiv{2}=\chi\left(\Phiv{1}\right)$.
%\begin{align*}
%N_2 = \begin{cases}
%\frac{2}{3}N_1^2+\frac{1}{3} & \text{ for the K}{}_1 \text{ closure.}\\
%\cfrac{5-2\sqrt{4-3N_1^2}}{3}& \text{ for the M}{}_1 \text{ closure.}
%\end{cases}
%\end{align*}
For more details see \cite{DubFeu99,Frank07,Monreal}.\\
\subsubsection{Mixed moments methods}
All mixed moment methods are solved using a kinetic scheme, see e.g. \cite{FraDubKla04}. It performs a trivial upwinding for the $"\pm"$-variables, respectively. The only interesting equation is the "full moment" equation for the density. Here the semi-discretized scheme at cell $j$ looks like
\begin{align*}
\de_t \psiv{0}_j + \dfrac{\left(\psiv{1}_{+,j}-\psiv{1}_{+,j-1}\right) +\left(\psiv{1}_{-,j+1}-\psiv{1}_{-,j}\right)}{\Delta x} = S\left(t,x,\gamma\right)_j
\end{align*}

\subsubsection{Time discretization and boundary conditions}
The time integration for all schemes is done using the Mathworks MATLAB \cite{MATLAB:2012} explicit adaptive second order Runge-Kutta integrator \emph{ode23}. The Fokker-Planck solution for the Source-Beam test case is calculated using the integrator for stiff equations \emph{ode23s}.
Note that both integrators mentioned here are not strongly stability preserving, see e.g. \cite{Gottlieb2001} for details.

Boundary conditions for moment problems are always problematic. We model them using ghost cells at the boundary where we directly prescribe the underlying kinetic distribution. From this we can consistently calculate the moments for all models of arbitrary order. 

Under some assumptions the spherical harmonics are equivalent to a special discrete ordinates Fokker-Planck discretization. With this, the spherical harmonics solution obeys the desired positivity for the distribution function. However, necessary for this are Mark boundary conditions \cite{Mar44,Mar45} which we do not use in our test cases. Therefore the spherical harmonics solution may oscillate into the negative as shown below.

\subsubsection{Laplace-Beltrami operator}
The Laplace-Beltrami operator is closed consistently using the corresponding distribution functions. The only exceptions are the mixed moment Kershaw closures. Using integration by parts gives
\begin{align*}
\intp{\mu^m&\de_\mu\left(\left(1-\mu^2\right)\de_\mu\psi\right)} = \underbrace{\left[\mu^m\left(\left(1-\mu^2\right)\de_\mu\psi\right)\right]_0^1}_{=0~\forall m\geq 1} - \intp{m\mu^{m-1}\left(1-\mu^2\right)\de_\mu\psi}\\
&= -\underbrace{\left[m\mu^{m-1}\left(1-\mu^2\right)\psi\right]_0^1}_{=0~\forall m\geq 2} + \intp{m\,\left(m - 1\right)  \mu^{m - 2}\psi - m\left(m + 1\right) \mu^m\psi}\\
&= m\,\left(m - 1\right)  \psip{m-2}- m\left(m + 1\right)\psip{m}
\end{align*}
and analogously
\begin{align*}
\intm{\mu^m&\de_\mu\left(\left(1-\mu^2\right)\de_\mu\psi\right)} = m\,\left(m - 1\right)  \psim{m-2}- m\left(m + 1\right)\psim{m}
\end{align*}
For $m=1$ we additionally get the microscopic term $\psi\left(0\right)$. This is obviously a problem since the Kershaw closure consists of Dirac delta functions. We therefore close this operator using the MP${}_n$ approximation.\\
\newline
\subsubsection{Realizability projection}
Sometimes the numerical approximations \\leave their realizability domain. This is a problem since then some closures cannot be evaluated anymore. The first order schemes used in this work preserve realizability since they are convex combinations of realizable vectors. 

Still problems may occur near the realizability boundary. This preserving property obviously holds only for exact arithmetic. If the evaluation of the flux is inexact (e.g. the tabulation for MM${}_1$ or simple numerical errors) the scheme may lose this property. Where necessary we apply a suitable projection back into the corresponding realizability domain. This is done by doing a line search along the ray $\alpha u + \left(1-\alpha\right) u_{eq}$ where $u$ is the current vector of moments and $u_{eq}$ is the equilibrium solution, that is, the moments of the constant distribution $\psi = \alpha$. Then we solve for $\alpha$ such that the corrected solution lies on the boundary (if evaluation on the boundary is possible, e.g. for Kershaw closures). This somehow corresponds to a limiting procedure in the angular variable, where our local (possibly nonlinear) ansatz is limited towards a constant solution.

\subsubsection{Comparison of methods}
For every example we present pictures as well as tables for comparison. In the tables, we always calculate for models ``$m1$'' and ``$m2$'' the difference in the corresponding $L^p$ sense for the densities $\psiv{0}$:
\begin{align*}
E_p(\psiv{0}_{{m1}},\psiv{0}_{{m2}})^p = \int\limits_{0}^{T}\int\limits_{D}\left|\psiv{0}_{{m1}}(x,t)-\psiv{0}_{{m2}}(x,t)\right|^p~dx~dt
\end{align*}
where $T$ is the final time of our calculations and $D$ the spatial domain. Characteristic errors are the $L^p$ difference at a specific time:
\begin{align*}
E_p^c(\psiv{0}_{{m1}},\psiv{0}_{{m2}},t)^p = \int\limits_{D}\left|\psiv{0}_{{m1}}(x,t)-\psiv{0}_{{m2}}(x,t)\right|^p~dx
\end{align*}

 \subsection{Beam in vacuum hitting an absorbing object}
In this test we model a beam hitting an object in vacuum. We therefore set for $x\in[-0.5,0.5]$
 \begin{gather*}
 \sigma_a(x) = \begin{cases}
 10 & \text{ if } x\in[-0.1,0.2]\\
 0 & \text{ else}
 \end{cases}, \hskip 1cm
 T(x) = Q(x)=0
 \end{gather*}
 and initial and boundary conditions
 \begin{align*}
 \psi\left(x,\mu,0\right) &= 10^{-4} &x\in(-0.5,0.5)\\
 \psi\left(-0.5,\mu>0,t\right) &= \frac{3\, \mathrm{e}^{3\, \mu + 3}}{\mathrm{e}^{6} - 1},~~~~~~~~&\psi\left(0.5,\mu<0,t\right) = 10^{-4}
 \end{align*}
 
Note that we don't choose a completely forward peaked solution (that is, a Dirac delta) because here all minimum entropy and Kershaw models are exact. As shown in Figure \ref{fig:1Beam} all models have some difficulties reproducing the exact shape of the Fokker-Planck solution for this specific choice of boundary data. One can nicely see the different wave packages arising from this Riemann problem at the left boundary. Fortunately all models recover the correct stationary solution.

Table \ref{tab:1Beam} confirms these observations. M${}_1$ and MM${}_1$ perform equally well. What is interesting is the large deviation between MM${}_1$ and MK${}_1$. As expected, MK${}_2$ provides better results as MK${}_1$. Going to a high number of modes (P${}_{51}$), the error becomes reasonably small.
   
   \InsertFig{
            \centering
                  \subfloat[$t=0.5$]{
                        \includegraphics[width = 0.48\textwidth]{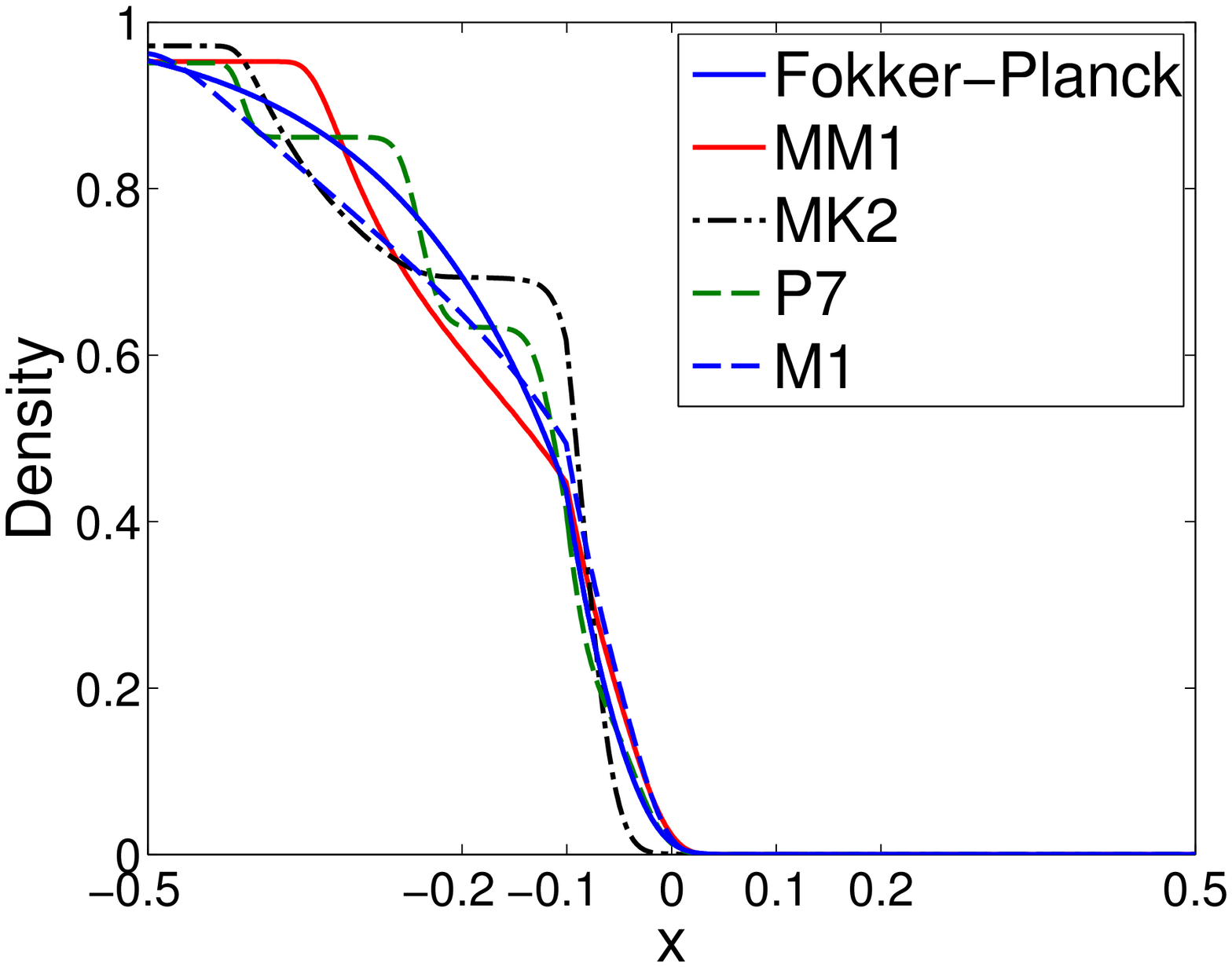}
                        \label{fig:1BeamCut1}
                        }
                        \subfloat[$t=1$]{
                                                \includegraphics[width = 0.48\textwidth]{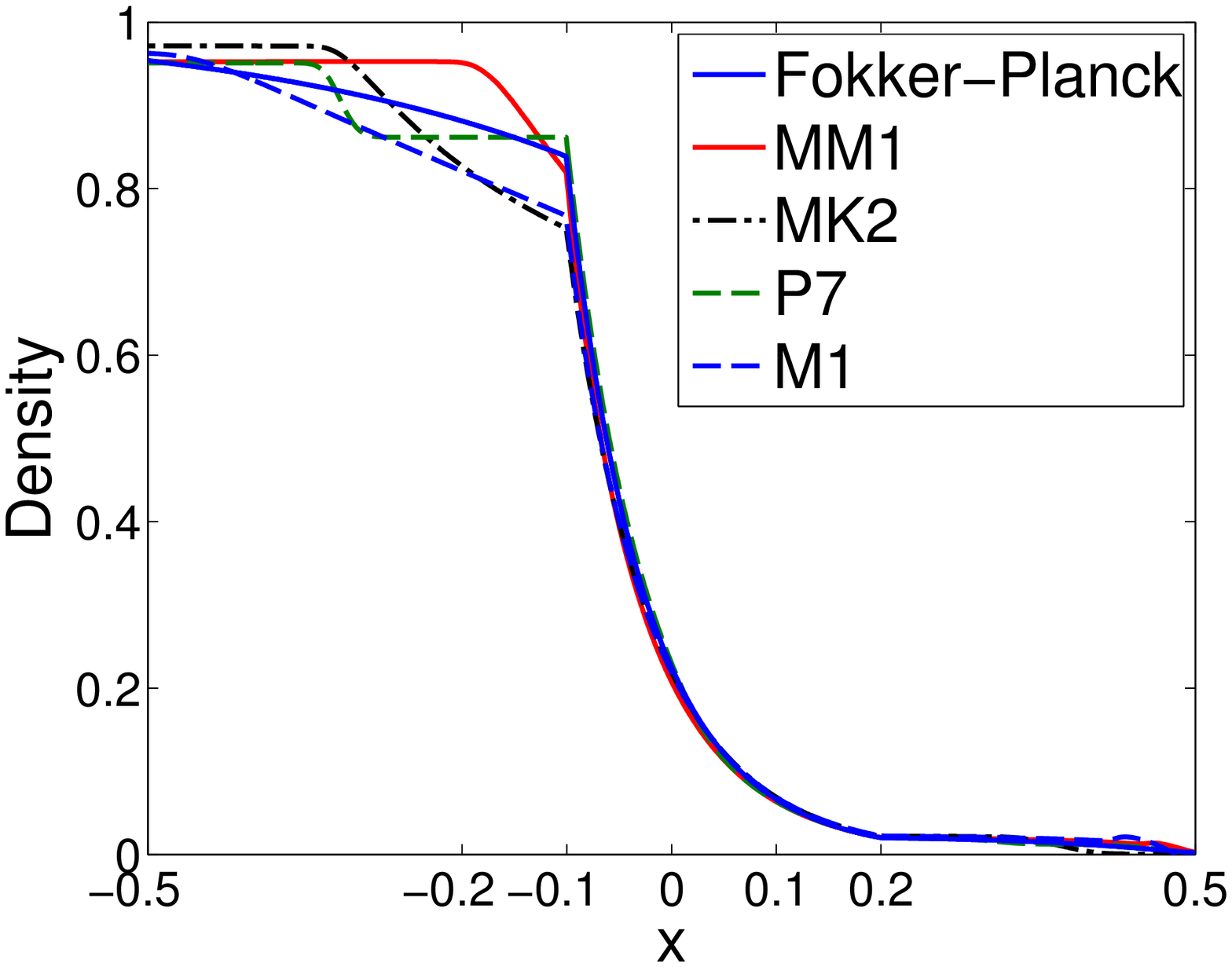}
                                                \label{fig:1BeamCut2}
                                                }\\
                  \subfloat[$t=2$]{
                                          \includegraphics[width = 0.48\textwidth]{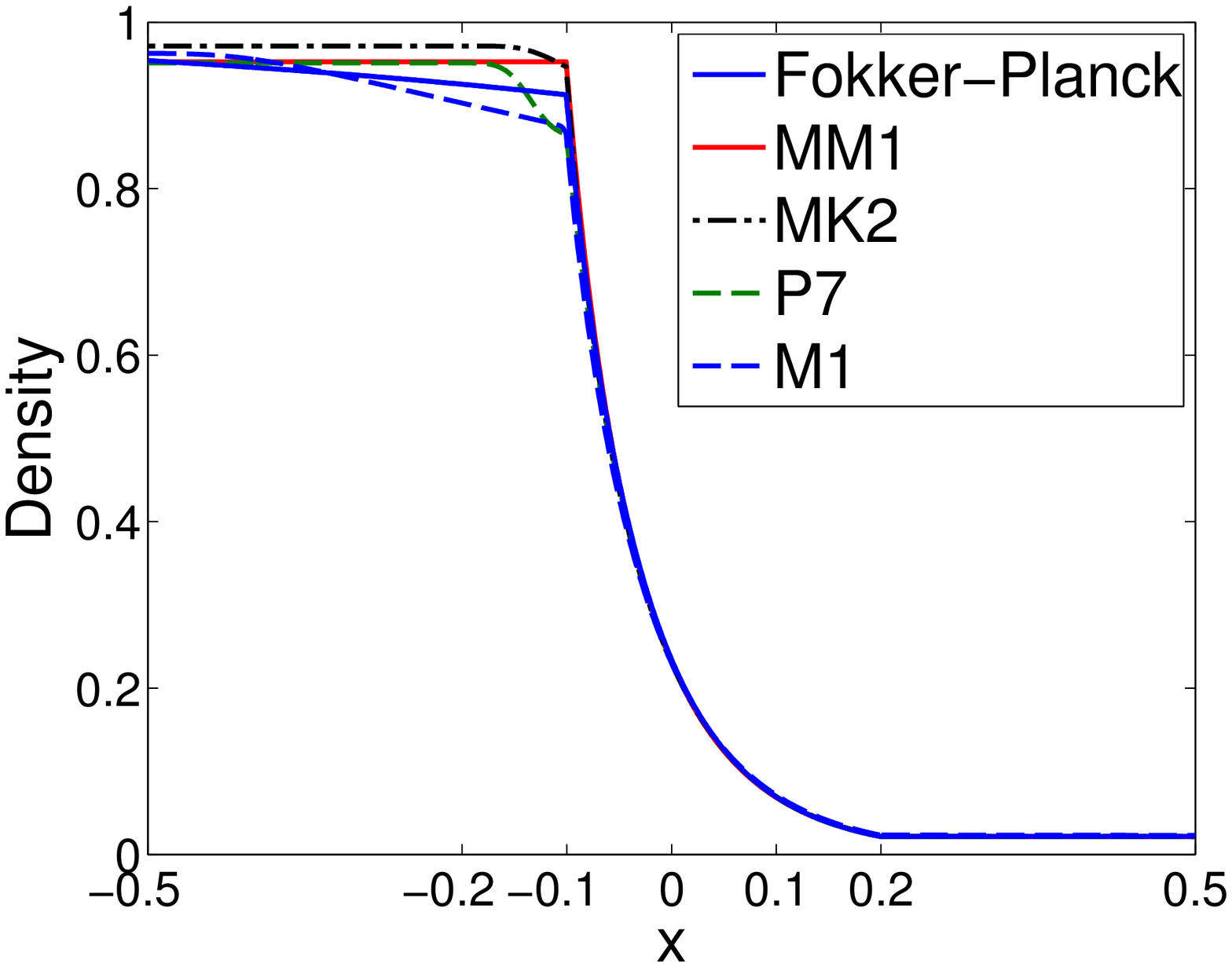}
                                          \label{fig:1BeamCut3}
                                          }
                                          \subfloat[$t=4$]{
                                                                  \includegraphics[width = 0.48\textwidth]{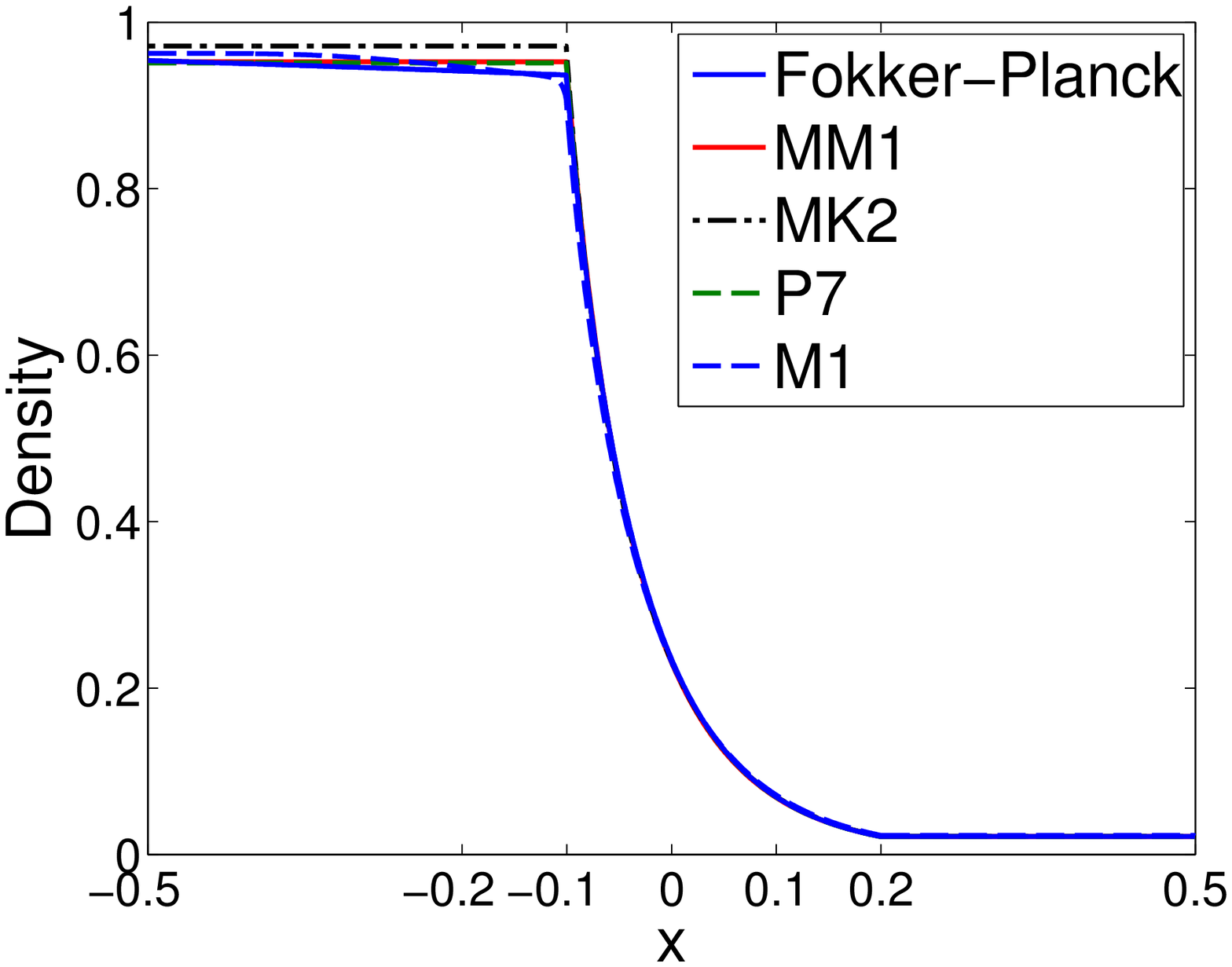}
                                                                  \label{fig:1BeamCut4}
                                                                  }\\                              
      	          }{Solutions of the One-Beam test case.}{\label{fig:1Beam}}

 \begin{table}[htbp]
   \centering
   \small\begin{tabular}{ccccccc}{Model} & {$L^1$} & {$L^2$} & {$L^\infty$} & {Char. $L^1$} & {Char. $L^2$} & {Char $L^\infty$}\\
  {MM${}_1$}  & 0.066339 & 0.033333 & 0.099496 & 0.01773 & 0.022299 & 0.041594\\
   {MK${}_1$} & 0.16669 & 0.071425 & 0.21744 & 0.049446 & 0.053638 & 0.075535\\
    {MK${}_2$}  & 0.10197 & 0.04527 & 0.19809 & 0.033389 & 0.037324 & 0.052803\\
     {M${}_1$}  & 0.07234 & 0.033757 & 0.11092 & 0.020511 & 0.023531 & 0.050847\\
     P${}_{7}$ & 0.042149 & 0.022403 & 0.0851 & 0.015638 & 0.020247 & 0.049295\\
          P${}_{51}$& 0.0059108 & 0.0026413 & 0.0055583 & 0.0023004 & 0.002608 & 0.0050124 
       \end{tabular}
   
     \caption{Relative $L^p$ errors, $p\in\{1,2,\infty\}$, for different models with respect to the Fokker-Planck solution $n_x = 1000$, $n_\mu = 800 $. One-Beam test case. Characteristic errors evaluated at $t=2$.}
     \label{tab:1Beam}
   \end{table}

\subsection{Two beams}
This test case models two beams entering into an absorbing medium. Nearly no particles are in the domain initially:
\begin{align*}
\psi\left(x,\mu,0\right) = 10^{-4} && x\in \left( -\frac{1}{2},\frac{1}{2} \right),
\end{align*}
and the equation is supplemented with boundary conditions
\begin{align*}
\psi\left(-\frac{1}{2},\mu>0,t\right) = 100\cdot\delta\left(\mu-1\right),&&\psi\left(\frac{1}{2},\mu<0,t\right) = 100\cdot\delta\left(\mu+1\right)
\end{align*}
Additionally we set the absorption parameter $\sigma_a = 4$ and the transport coefficient $T = 0$. This is the classical setting where full moment $M_1$ (and $K_1$ as well) fails due to a zero netflux during the collision of the two beams. This is shown in Figure \ref{fig:2Beams1}. Error estimates for the different models are shown in the Table \ref{tab:2Beams}.

\begin{table}[htbp]
\centering
\small\begin{tabular}{ccccccc} {Model} & {$L^1$} & {$L^2$} & {$L^\infty$} & {Char. $L^1$} & {Char. $L^2$} & {Char $L^\infty$}\\
{MM${}_1$} & 0.0001528 & 0.0039495 & 0.0055959 & 0.0038968 & 0.0042576 & 0.0055318\\ 
{MK${}_1$} & 0.00010846 & 0.0025032 & 0.0025472 & 0.0025031 & 0.0025032 & 0.0025102\\ 
{MK${}_2$} & 0.00010816 & 0.0025029 & 0.0044977 & 0.0024997 & 0.0025043 & 0.0027055\\ 
 {MP${}_2$} & 0.010856 & 0.39845 & 2.0879 & 0.22026 & 0.38489 & 2.087\\ 
 {MP${}_5$} & 0.0041039 & 0.21726 & 1.5572 & 0.071395 & 0.20711 & 1.5519\\ 
  {MP${}_{10}$} & 0.0015805 & 0.11101 & 0.9418 & 0.023389 & 0.10519 & 0.93944\\ 
   {M${}_1$} & 0.0057931 & 0.18592 & 0.25789 & 0.16074 & 0.19928 & 0.21276\\ 
  {K${}_1$} & 0.0055333 & 0.17446 & 0.25725 & 0.15275 & 0.1857 & 0.21545\\
 {P${}_5$} & 0.0035746 & 0.10779 & 0.29795 & 0.062191 & 0.09103 & 0.23609\\ 
 {P${}_{11}$} & 0.0015977 & 0.056482 & 0.16725 & 0.025458 & 0.047333 & 0.16291\\ 
 {P${}_{21}$} & 0.0006984 & 0.030108 & 0.12362 & 0.010751 & 0.025912 & 0.12202 
      \end{tabular}
\caption{Relative $L^p$ errors, $p\in\{1,2,\infty\}$, for different models with respect to the Fokker-Planck solution with $n_x =1000 $, $ n_\mu = 800$. Two beam test case.  Characteristic errors evaluated at $t=4$.}
\label{tab:2Beams}
\end{table}\normalsize

Note that in this case the full moment P${}_n$ model converges faster towards the Fokker-Planck solution than the mixed MP${}_n$ model. Here we always compare approximately the same number of variables (e.g. P${}_{21}$ with MP${}_{10}$ to avoid instabilities in the P${}_n$ model). Mixed minimum-entropy (MM${}_1$) and mixed Kershaw closures (MK${}_1$ and MK${}_2$) perform well.

As shown in Figure \ref{fig:2Beams1} the MM${}_1$ solution is indistinguishable from the Fokker-Planck solution while M${}_1$ and P${}_5$ behave differently. The $L^1$ errors in Table \ref{tab:2Beams} show that (up to numerical deviations of order $0.1h = \frac{1}{10n_x}$) MM${}_1$, MK${}_1$ and MK${}_2$ give the same results as Fokker-Planck. This is to be expected since the Fokker-Planck solution is a linear combination of two Dirac deltas in $\mu$ which can be arbitrarily closely prescribed by the MM${}_1$ model and exactly prescribed by MK${}_1$ and MK${}_2$.

\InsertFig{\subfloat[Fokker-Planck]{
      \includegraphics[width = 0.48\textwidth]{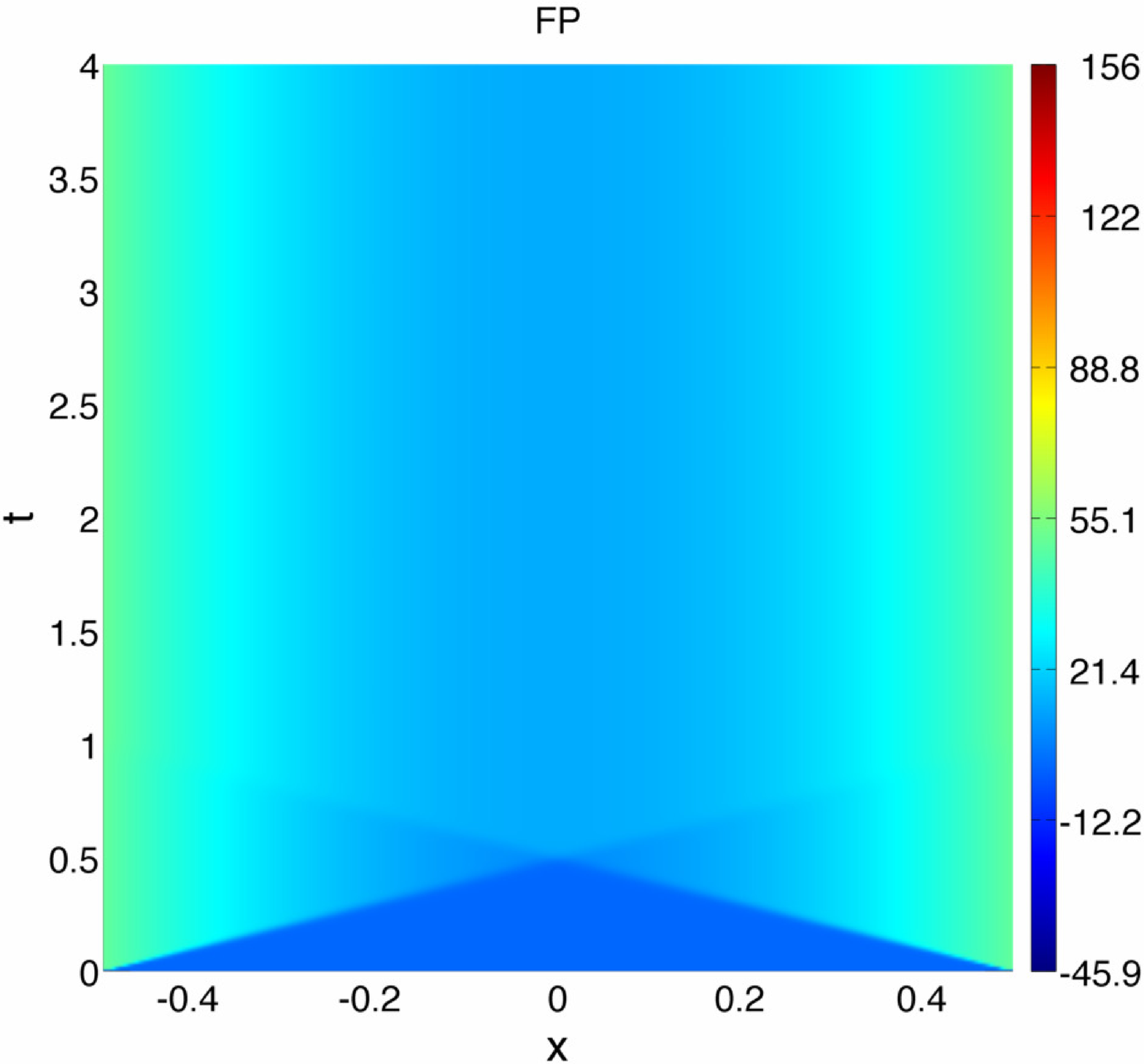}
      \label{fig:2Beams-FP}
      }
\subfloat[M${}_1$]{
      \includegraphics[width = 0.48\textwidth]{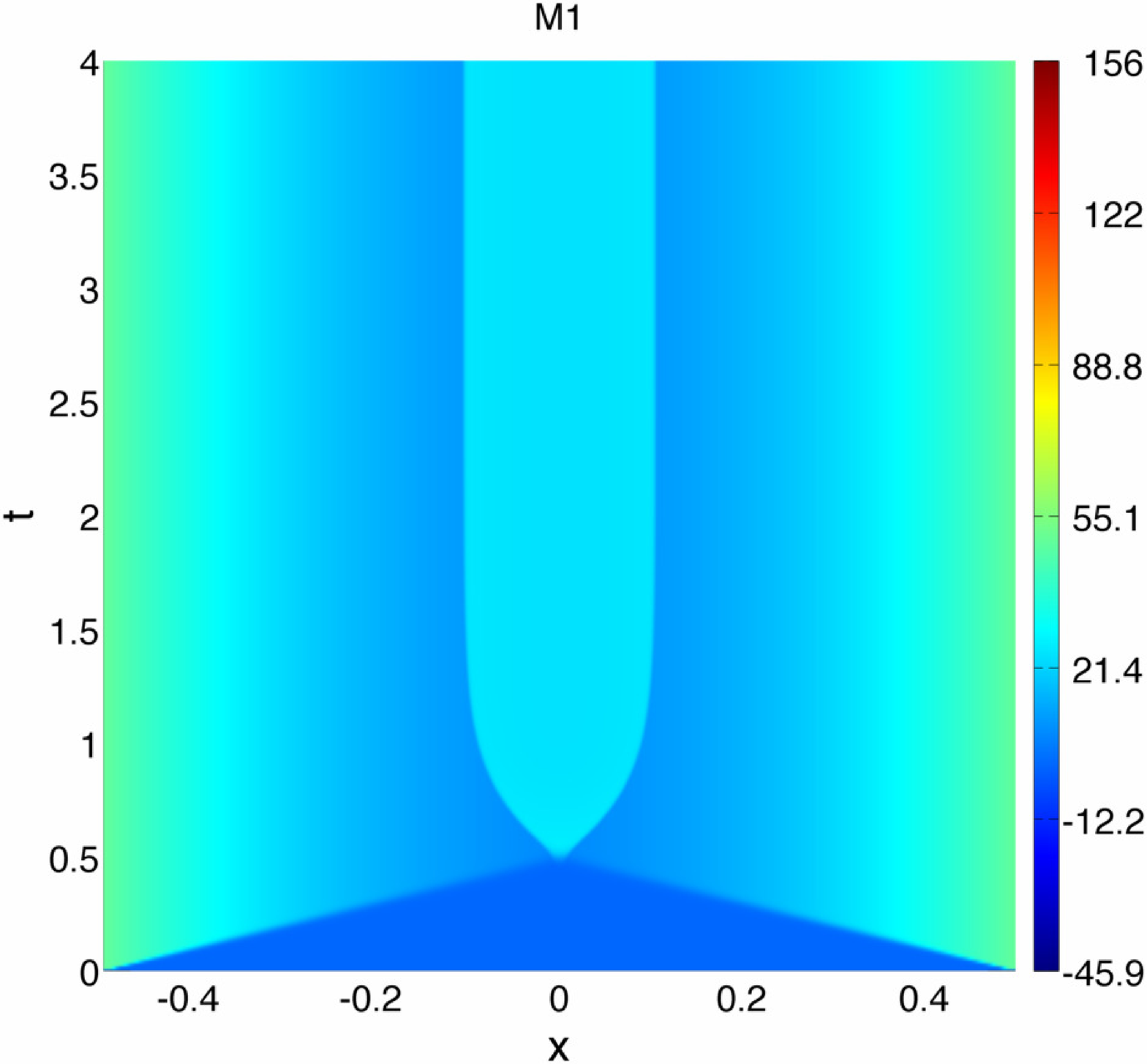}
               \label{fig:2Beams-M1}}\\
               \centering
               \subfloat[$t = 4$]{
                     \includegraphics[width = 0.9\textwidth]{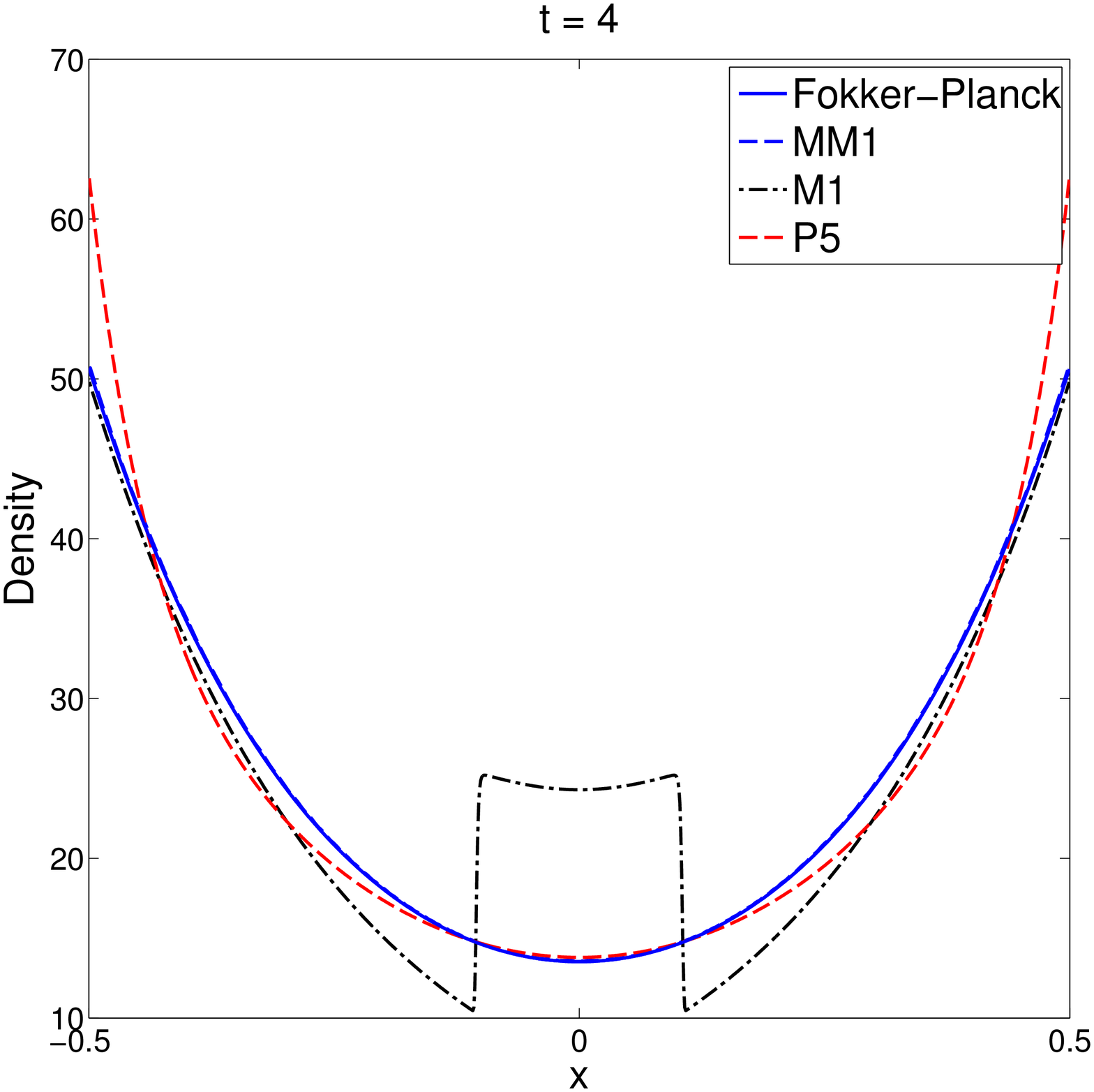}
                     \label{fig:2BeamsCut}
                     }
   	          }{Solutions of the Two-Beams test case. MM${}_1$ is approximately identical to Fokker-Planck.}{\label{fig:2Beams1}}

  \subsection{Rectangular IC}
   In this test case we start with an isotropic distribution where nearly all mass is concentrated in the middle of the domain $X = [0,7]$:
  \begin{align*}
  \psi\left(x,\mu,0\right) = \begin{cases}
  10&\text{ if } x\in [3,4]\\
  10^{-4}&\text{ else}
  \end{cases}
  \end{align*}
  At the boundary we have
  \begin{align*}
  \psi\left(0,\mu>0,t\right) = 10^{-4},&&  \psi\left(7,\mu<0,t\right) = 10^{-4}
  \end{align*}
    We use a slightly scattering material without absorption, therefore $\sigma_a = 0$ and $T = 10^{-2}$.
  \begin{table}[htbp]
  \centering
  \small\begin{tabular}{ccccccc}{Model} & {$L^1$} & {$L^2$} & {$L^\infty$} & {Char. $L^1$} & {Char. $L^2$} & {Char $L^\infty$}\\
{MM${}_1$} & 0.019549 & 0.23133 & 0.16337 & 0.16749 & 0.17692 & 0.2757\\ 
{MK${}_1$} & 0.022142 & 0.28221 & 0.3131 & 0.2479 & 0.27771 & 0.62157\\ 
{MK${}_2$} & 0.016665 & 0.25969 & 0.22311 & 0.096056 & 0.11219 & 0.22806\\ 
{MP${}_5$} & 0.0027983 & 0.038741 & 0.025006 & 0.012051 & 0.013356 & 0.021774\\
 {MP${}_{10}$} & 0.00017611 & 0.0028254 & 0.0024361 & 0.0019917 & 0.0021764 & 0.0036513\\
  {M${}_1$} & 0.03298 & 0.39211 & 0.22989 & 0.14476 & 0.14726 & 0.17734\\
   {P${}_{11}$} & 0.0037782 & 0.051546 & 0.0377 & 0.022936 & 0.029121 & 0.055617\\
    {P${}_{21}$} & 0.00020889 & 0.0047516 & 0.0090014 & 0.0040468 & 0.0058232 & 0.016255 
          \end{tabular}
  \caption{Relative $L^p$ errors, $p\in\{1,2,\infty\}$, for different models with respect to the Fokker-Planck solution with $n_x = 1000$, $n_\mu = 800 $. Rectangular IC test case. Characteristic errors evaluated at $t=1$.}
  \label{tab:RectIC}
  \end{table}
  
  \InsertFig{\subfloat[Fokker-Planck]{
        \includegraphics[width = 0.43\textwidth]{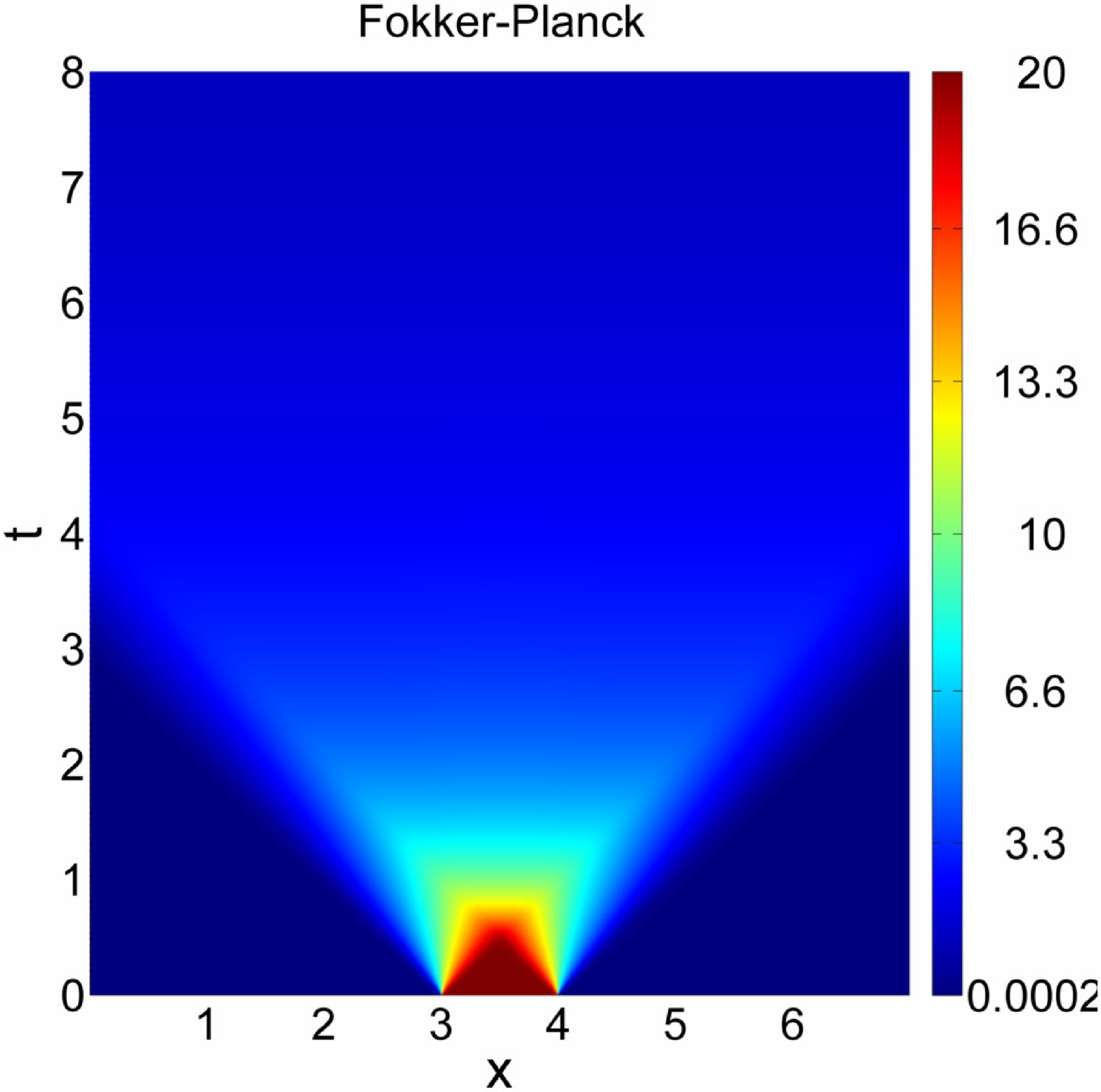}
        \label{fig:RectIC-FP}
        }
  \subfloat[MM${}_1$]{
        \includegraphics[width = 0.43\textwidth]{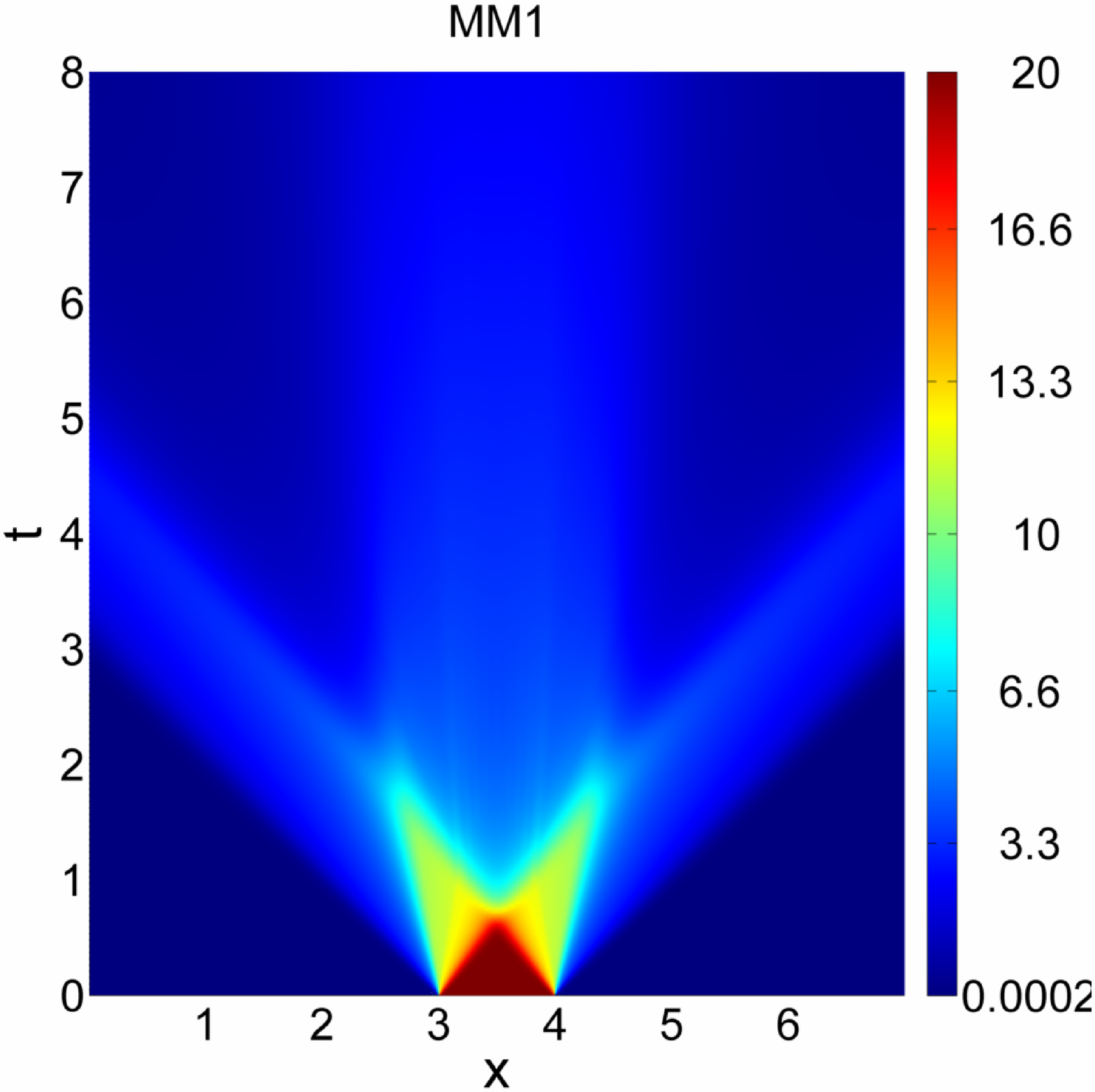}
                 \label{fig:RectIC-MM1}
        }\\
        \subfloat[MK${}_1$]{
                 \includegraphics[width = 0.43\textwidth]{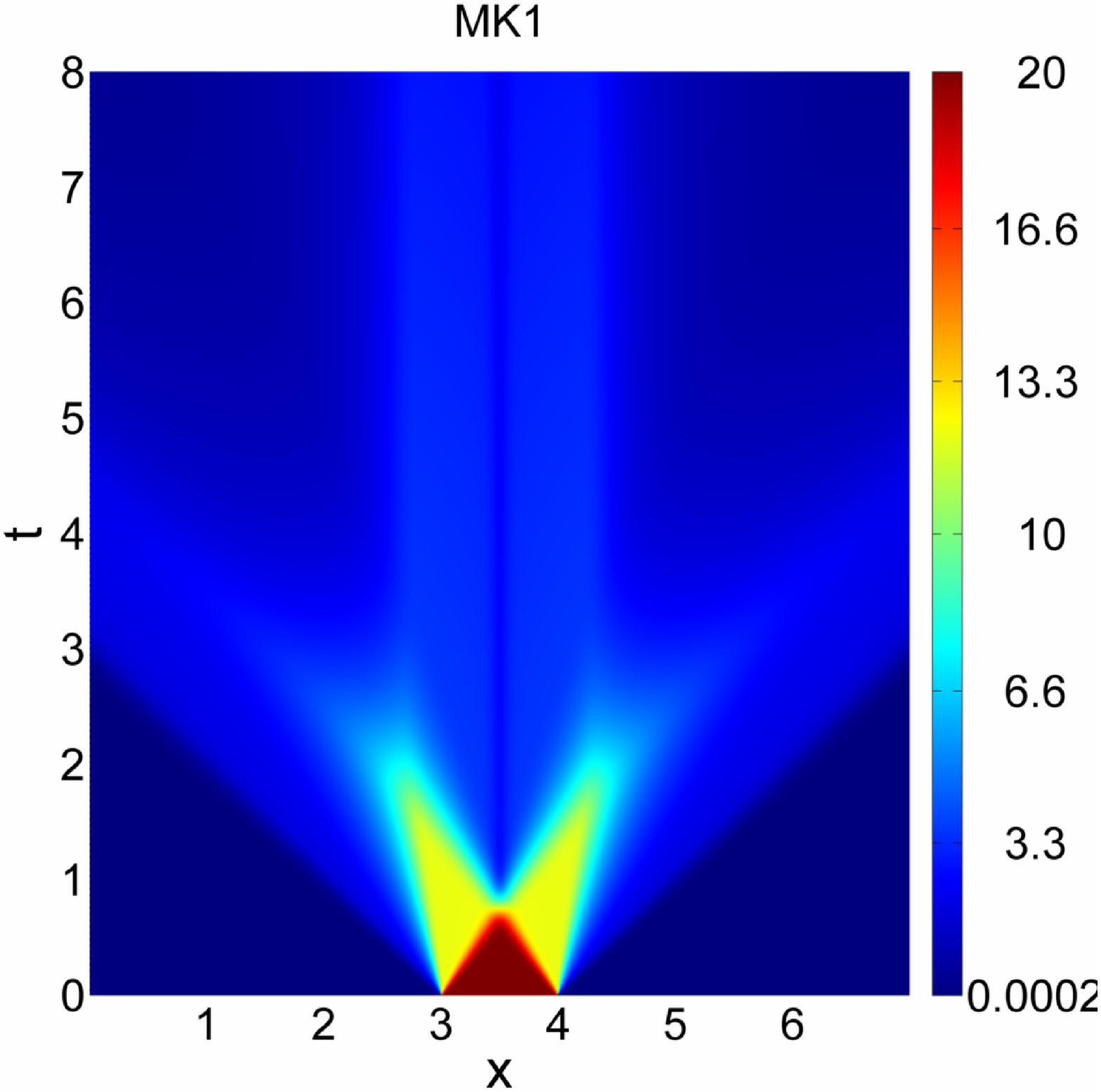}
                 \label{fig:RectIC-MK1}
                 }
        \subfloat[MK${}_2$]{
                 \includegraphics[width = 0.43\textwidth]{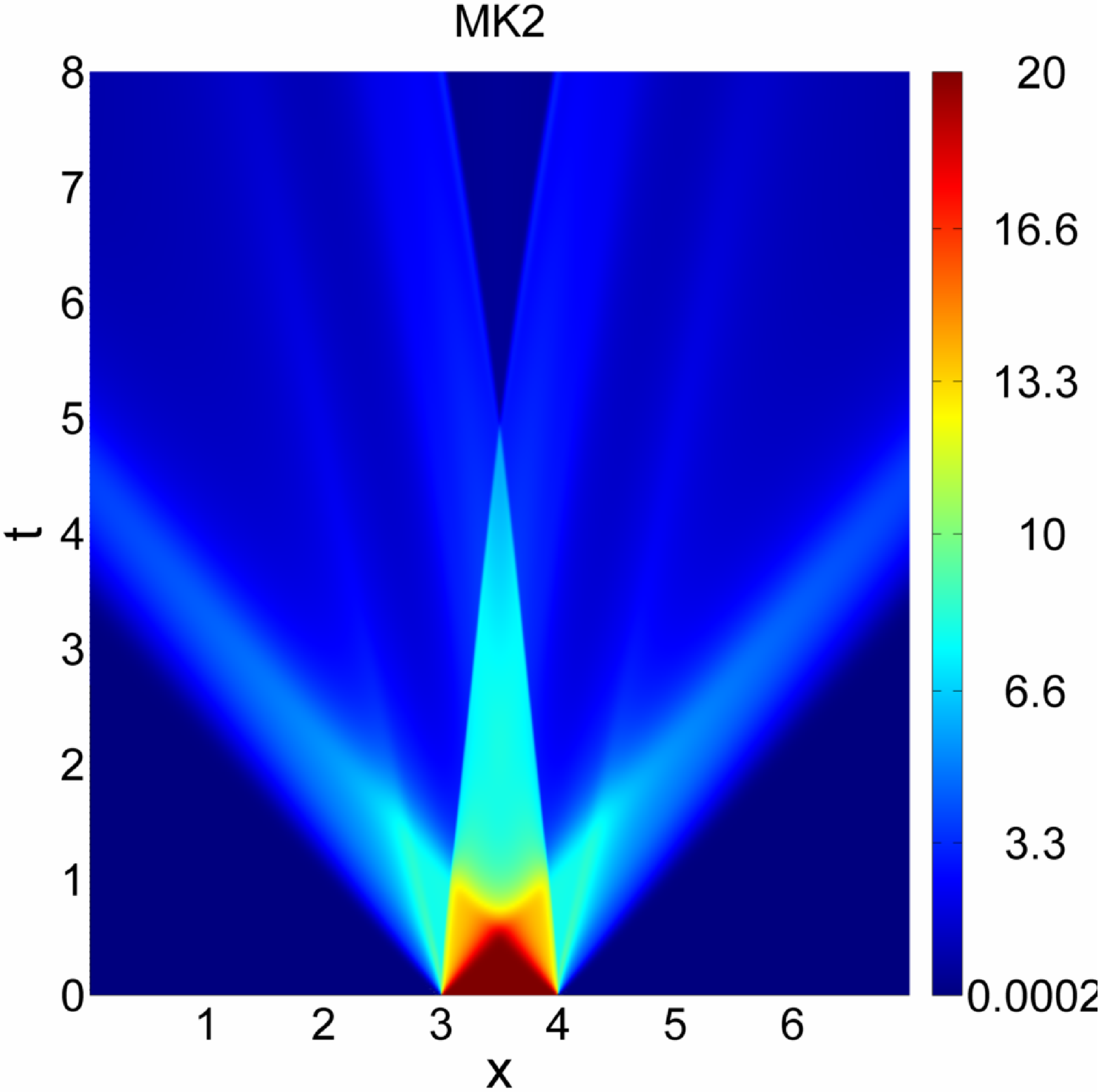}
                          \label{fig:RectIC-MK2}
                 }\\
    	 \subfloat[M${}_1$]{
  \includegraphics[width = 0.43\textwidth]{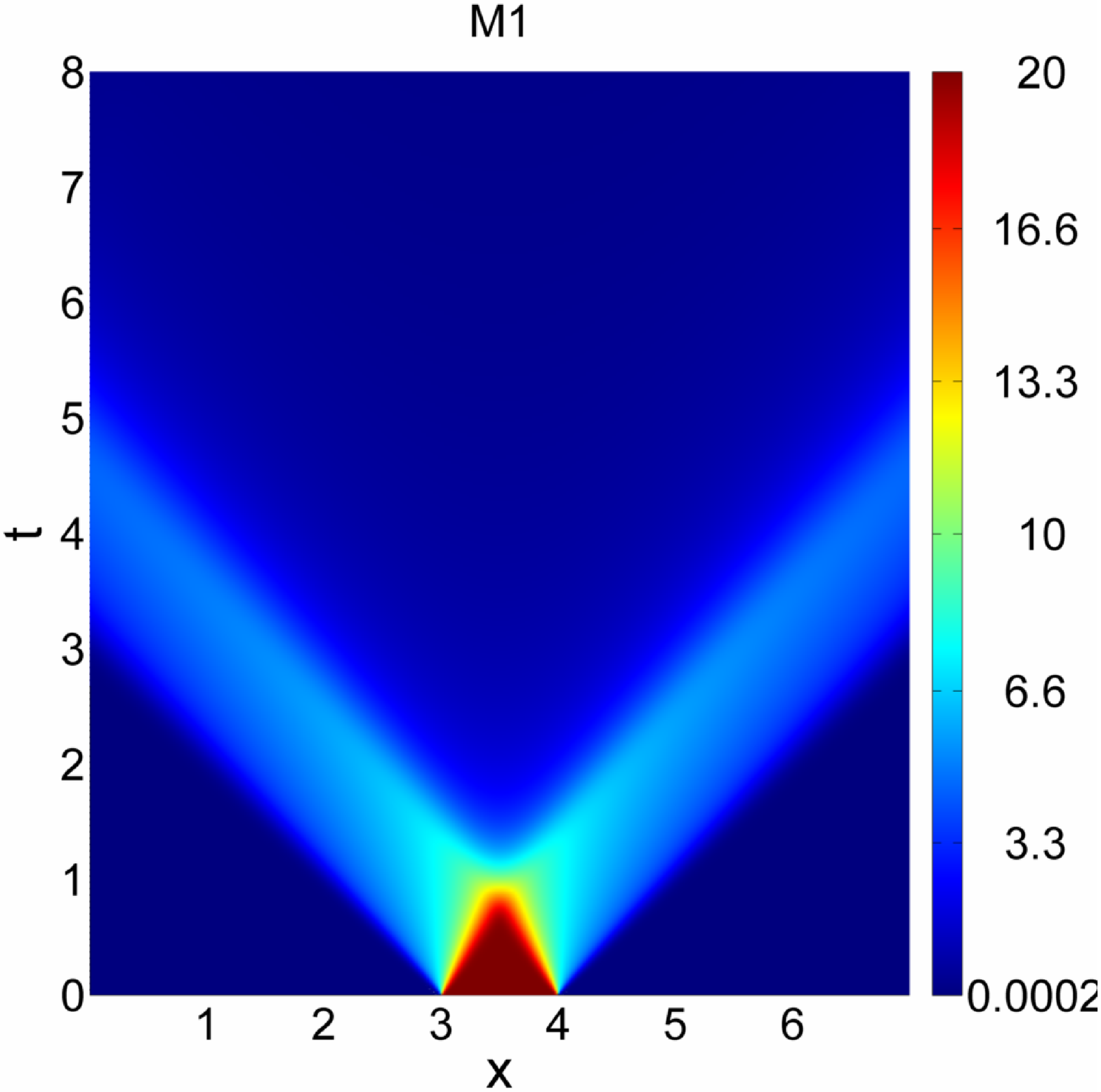}
        \label{fig:RectIC-M1}}
        \subfloat[MP${}_{10}$]{
        \includegraphics[width = 0.43\textwidth]{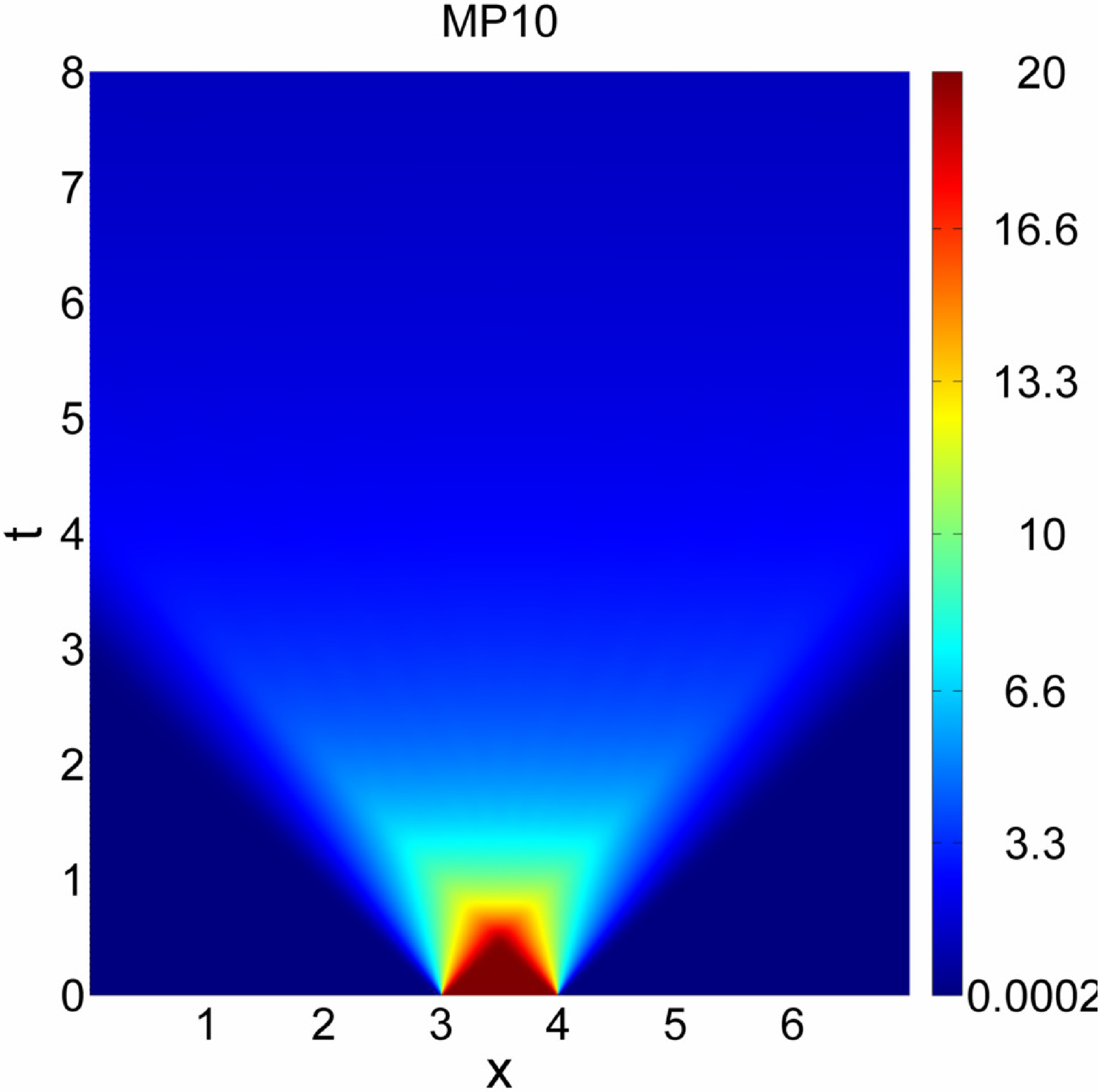}
                 \label{fig:RectIC-MP10}}\\
%                  \subfloat[P${}_{21}$]{
%                 \includegraphics[width = 0.48\textwidth]{RectIC/RectIC-P21}
%                          \label{fig:RectIC-P20}}
     	          }{Solutions of the Rectangular IC test case.}{\label{fig:RectIC1}}
  
     \InsertFig{
              \centering
                    {
                          \includegraphics[width = 0.9\textwidth]{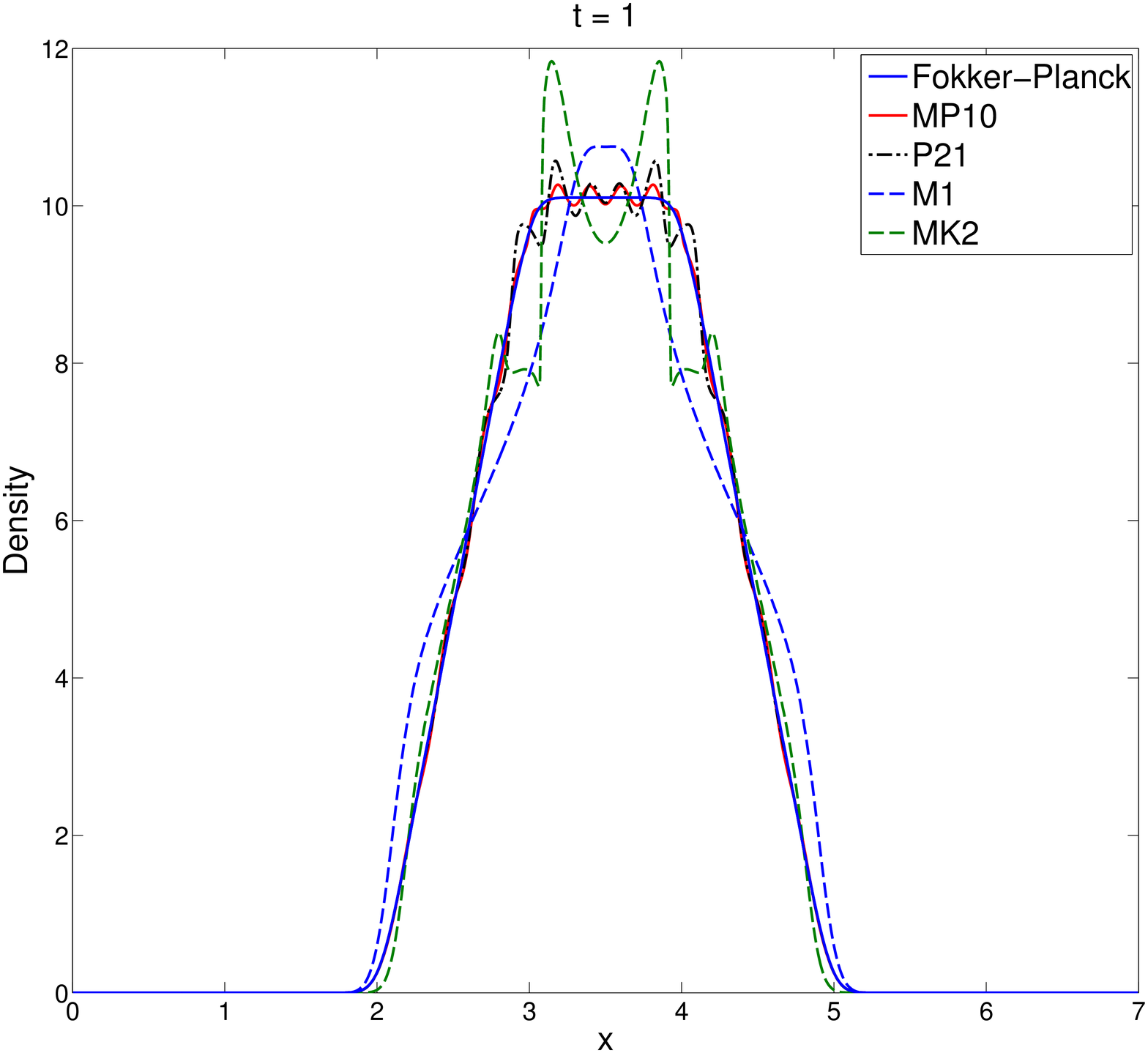}
                          \label{fig:RectICCut}
                          }\\
        	          }{Cut of the solutions of the Rectangular IC test case evaluated at the characteristical time $t=1$.}{\label{fig:RectIC}}

 Here, the mixed polynomial MP${}_n$ models perform better than the full spherical harmonics P${}_n$ models. MM${}_1$ performs slightly better than MK${}_1$ but worse than MK${}_2$. This can be expected since more waves for the solution are available. By the same argument the M${}_1$ model performs worse than the other non-linear models. 
 As one can see in the space-time plots in Figure \ref{fig:RectIC1} the non-linear closures are exact at the beginning for a short period of time. This is true until the Fokker-Planck distribution is no longer isotropic.
 In Figure \ref{fig:RectIC} the different wave packages of the models can be seen. 
 \subsection{Source-Beam}
 This test has been used in \cite{Frank:1471763}. However we do not use the smoothed version but the discontinuous one with $X = [0,3]$,
 \begin{gather*}
 \sigma_a(x) = \begin{cases}
 1 & \text{ if } x\leq 2\\
 0 & \text{ else}
 \end{cases}, \hskip 1cm
 T(x) = \begin{cases}
  0 & \text{ if } x\leq 1\\
  2 & \text{ if } 1<x\leq 2\\
  10 & \text{ else}
  \end{cases},\\
   Q(x) = \begin{cases}
    1 & \text{ if } 1\leq x\leq 1.5\\
    0 & \text{ else}
    \end{cases}\\
 \end{gather*}
 and initial and boundary conditions
 \begin{gather*}
 \psi\left(x,\mu,0\right) = 10^{-4}\\
 \psi\left(0,\mu>0,t\right) = \delta\left(\mu-1\right),~~~~~~~~\psi\left(3,\mu<0,t\right) = 10^{-4}
 \end{gather*}
 
 As shown in Figure \ref{fig:Sourcebeam} the full moment models are not able to reproduce the Fokker-Planck solution, not even in steady state ($t=4$). The MK${}_2$ model provides reasonably better results than MM${}_1$ and MK${}_1$.
   
   \InsertFig{
            \centering
                  \subfloat[$t=0.5$]{
                        \includegraphics[width = 0.48\textwidth]{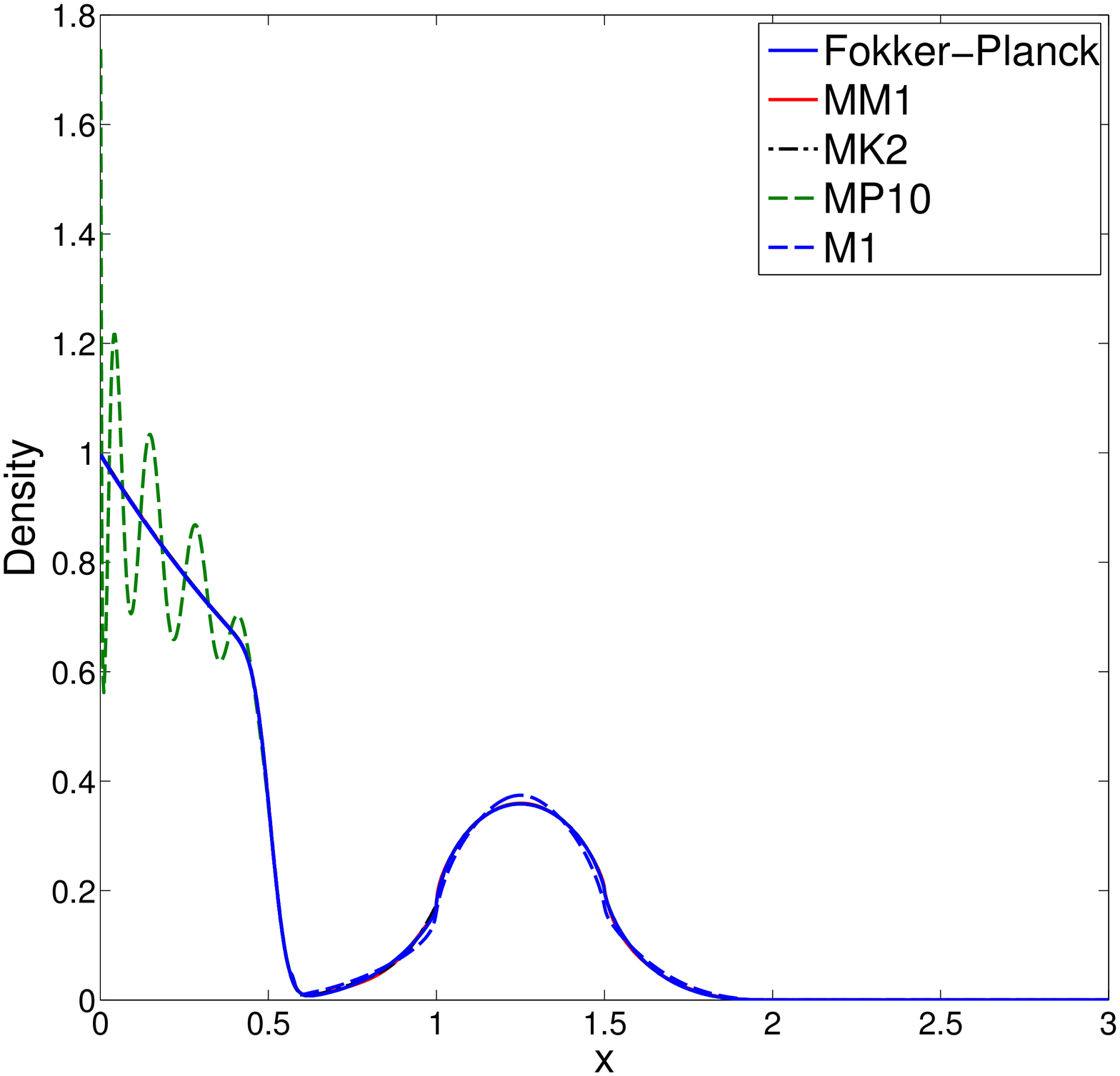}
                        \label{fig:SourcebeamCut1}
                        }
                        \subfloat[$t=1$]{
                                                \includegraphics[width = 0.48\textwidth]{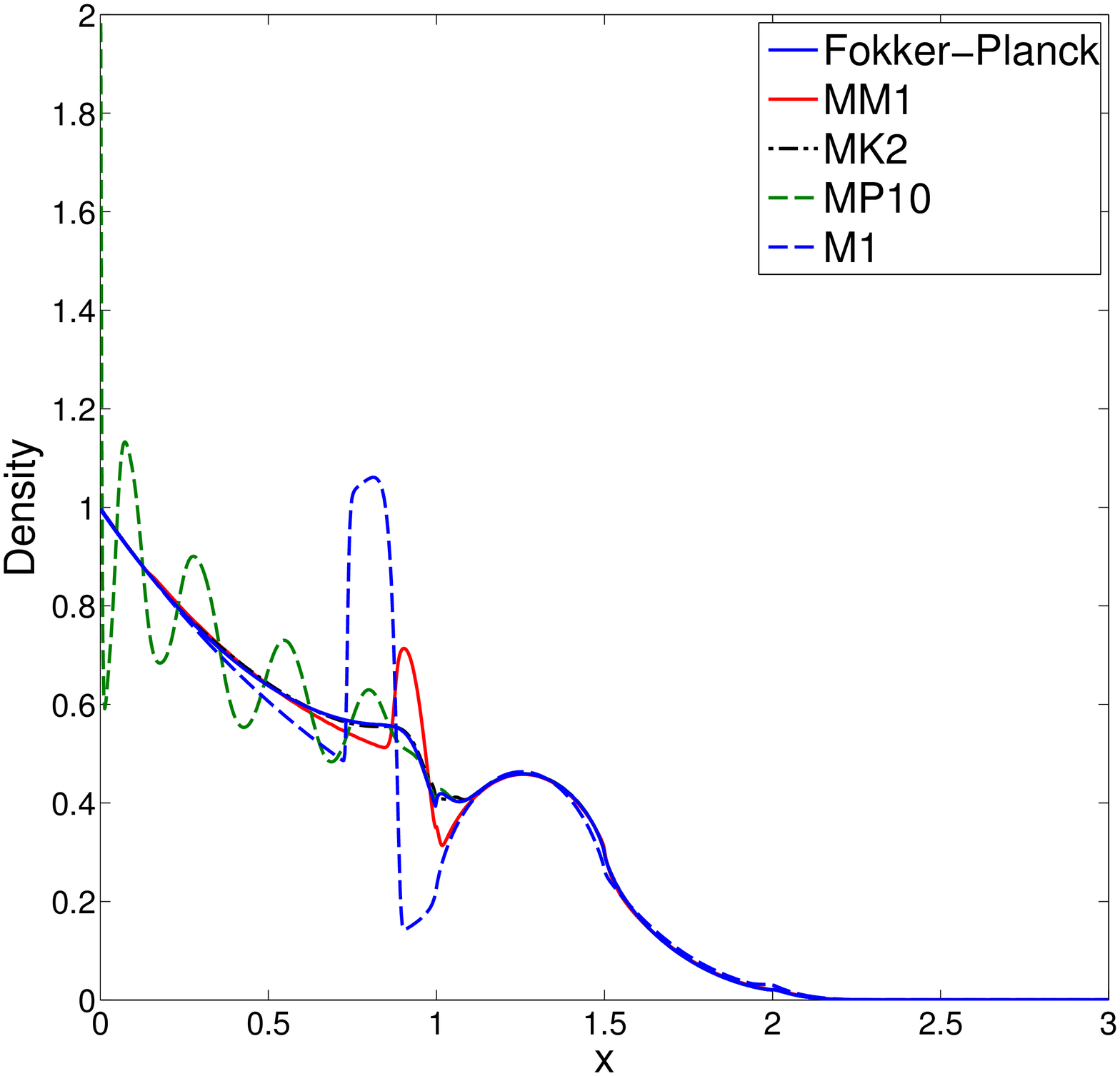}
                                                \label{fig:SourcebeamCut2}
                                                }\\
                  \subfloat[$t=2$]{
                                          \includegraphics[width = 0.48\textwidth]{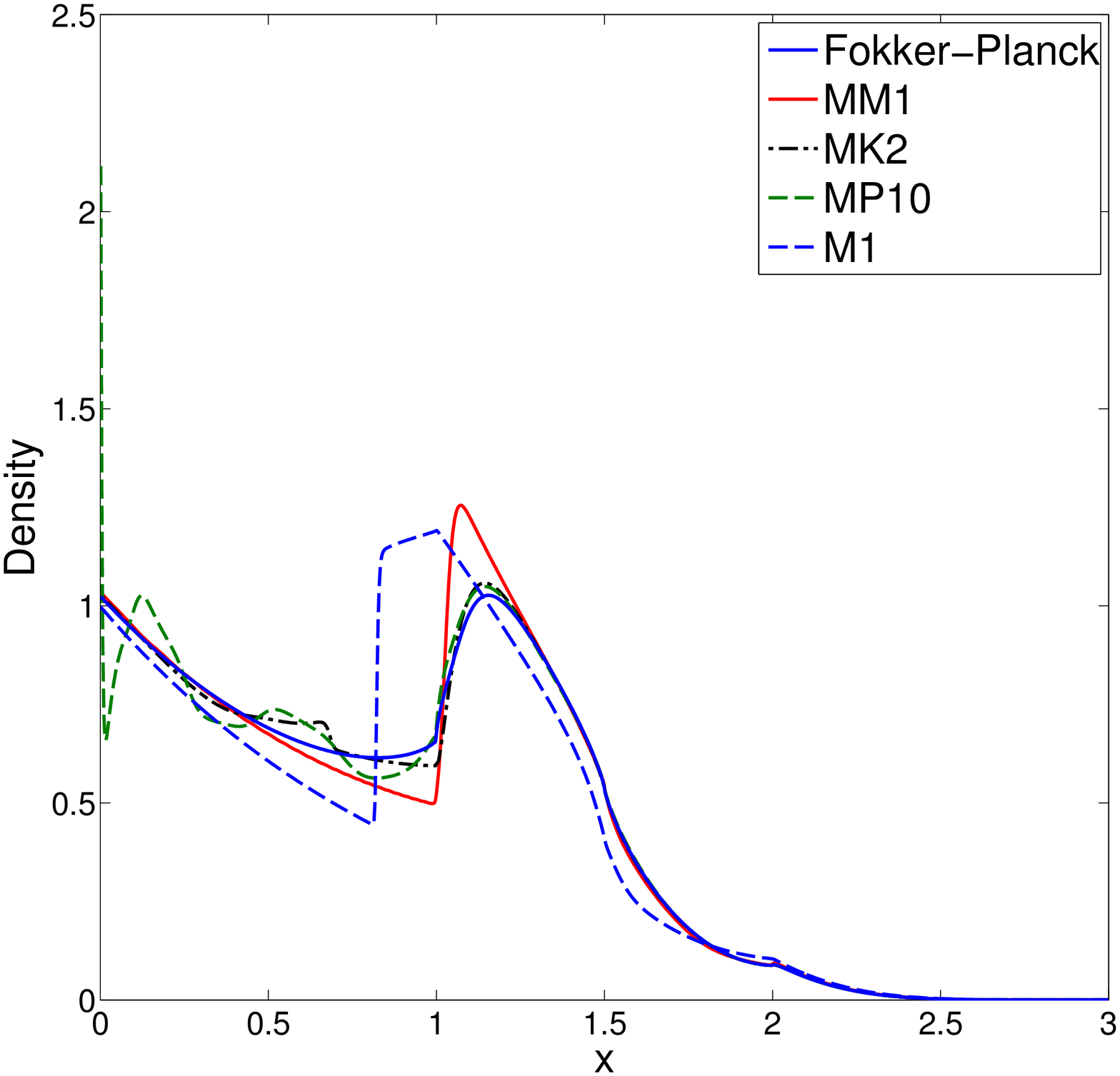}
                                          \label{fig:SourcebeamCut3}
                                          }
                                          \subfloat[$t=4$]{
                                                                  \includegraphics[width = 0.48\textwidth]{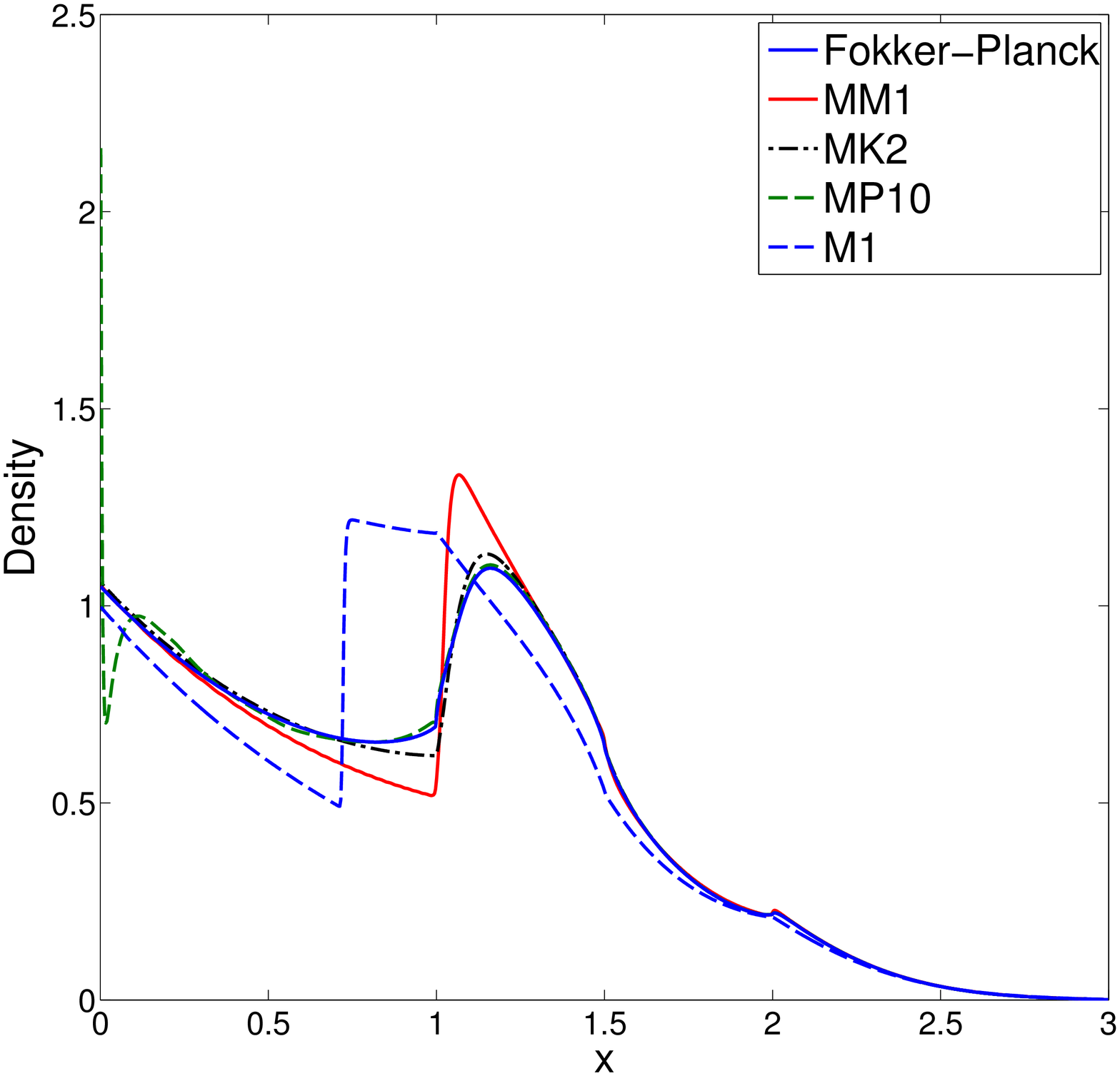}
                                                                  \label{fig:SourcebeamCut4}
                                                                  }\\                              
      	          }{Solutions of the Source-Beam test case.}{\label{fig:Sourcebeam}}
      	          
 \begin{table}[htbp]
   \centering
   \small\begin{tabular}{ccccccc}{Model} & {$L^1$} & {$L^2$} & {$L^\infty$} & {Char. $L^1$} & {Char. $L^2$} & {Char $L^\infty$}\\
  {MM${}_1$} & 0.043001 & 0.1082 & 0.39092 & 0.061205 & 0.10897 & 0.34317\\
   {MK${}_1$} & 0.048038 & 0.10979 & 0.40676 & 0.063409 & 0.10777 & 0.3108\\
    {MK${}_2$} & 0.013482 & 0.029286 & 0.10276 & 0.021587 & 0.033466 & 0.10427\\
     {MP${}_{10}$} & 0.04345 & 0.102 & 1.0156 & 0.050784 & 0.098315 & 1.0645\\ 
     {M${}_1$} & 0.13797 & 0.27406 & 0.5111 & 0.1879 & 0.26932 & 0.53131\\
     P${}_{21}$ & 0.018499 & 0.032356 & 0.11481 & 0.020531 & 0.029763 & 0.12231 
       \end{tabular}
   
     \caption{Relative $L^p$ errors, $p\in\{1,2,\infty\}$, for different models with respect to the Fokker-Planck solution $n_x = 1000$, $n_\mu = 800 $. Source-Beam test case. Characteristic errors evaluated at $t=2$.}
     \label{tab:Sourcebeam}
   \end{table}
 \section{Conclusions and outlook}
 \label{sec:CON}
We have established a theory for a mixed moment realizability approach in one space dimension. The resulting minimum entropy models as well as the corresponding Kershaw closures perform well in the numerical tests that we made. The zero-net-flux problem of full moment minimum entropy and Kershaw closures can be overcome. Additionally a consistent approximative model for the Fokker-Planck equation with Laplace-Beltrami operator has been derived.\\

Using the techniques provided in this paper an approximation using arbitrary numbers of moments can be derived.

%This approximative MK${}_n$ closure will converge for $n\to\infty$ towards the Fokker-Planck solution while the underlying density distribution function is always positive.
Compared with the corresponding minimum entropy model (which provides the same advantages) the MK${}_n$ models can be evaluated much more efficiently since the closure can be evaluated analytically without the need of many nonlinear solves. Thus it can compete in speed with the corresponding polynomial MP${}_n$ model.

We want to emphasize that it is in principle possible to use the QMOM-approach not only for full moments but as well for other partial or mixed moments. This can be done by simply plugging in $\psi_{QMOM_N}$ in the corresponding moment problem ensuring that the atoms are within the right support of the measure (e.g. for half moments in $[-1,0]$ and $[0,1]$ respectively). The influence of the moment problem on the atoms and densities in the QMOM-algorithm is to the authors' knowledge not investigated yet and may be subject to future research.

Until now there exists no general realizability theory for the moment problems in three dimensions. Explicit necessary and sufficient conditions have only been provided for up to $n=2$ \cite{Ker76}. 

A lot of work has been done to derive similar partial-moment models in $3D$. One general approach is the minimum-entropy quarter-moment approximation \cite{frank2005partial} which performs well for the transport equation. However, it has the same problems as the half-moment approximation in one dimension, namely that the Fokker-Planck operator must be closed consistently. The typical strategy which leads to decoupling provides unphysical results and does not reproduce the correct solution.
Again a formulation using a continuous distribution function is expected to provide reasonably good results in many situations. However, realizability theory for these mixed moments with $n > 2$ still has to be developed. We assume that the procedure works as in this paper, but without a general theory for partial moments it seems hard to formulate the correct conditions.\\

Mixed minimum-entropy moments of arbitrary order suffer from the problem of numerical inversion of the closure, similar to full moment M${}_n$ \cite{Hauck2010}. Therefore mixed Kershaw closures (which can be closed analytically) are of general interest because they can be evaluated with the same effort as spherical harmonics while guaranteeing the realizability (and especially the positivity) of the solution.

\bibliographystyle{siam}
\bibliography{RadLit}

\end{document}